\documentclass{lmcs} 
\usepackage[utf8]{inputenc}
\pdfoutput=1

\usepackage{booktabs}

\usepackage{lastpage}
\lmcsdoi{18}{1}{37}
\lmcsheading{}{\pageref{LastPage}}{}{}%
{Mar.~05,~2019}{Mar.~03,~2022}{}


\theoremstyle{plain} 

\newcommand{\A}{\mathcal{A}}

\newcommand{\Aut}{\ensuremath{\mathcal{A}}}

\newcommand{\In}{\mathcal{I}}
\newcommand{\Out}{\mathcal{O}}
\newcommand{\Barr}{\mathcal{B}}

\newcommand{\vb}{\vec{b}}

\begin{document}

	\title[Sequential Relational Decomposition]{Sequential Relational Decomposition}

\titlecomment{An extended abstract of this article appeared in the Proceedings of the 33rd Annual {ACM/IEEE} Symposium on Logic in Computer Science, {LICS} 2018.}

	\author[D.~Fried]{Dror Fried\rsuper{a}}	
	\address{Department of Computer Science, The Open University of Israel, Israel}	
	\email{dfried@openu.ac.il}  

	\author[A.~Legay]{Axel Legay\rsuper{b}}	
	\address{Department of Computer Science, Université Catholique de Louvain, Belgium}	
	\email{ axel.legay@uclouvain.be}  

	\author[J.~Ouaknine]{Jo\"{e}l Ouaknine\rsuper{c,d}}	
	\address{Max Planck Institute for Software Systems, Saarland Informatics Campus, Germany}
	\address{Department of Computer Science, {Oxford University, UK}}
	\email{joel@mpi-sws.org}	

	\author[M.Y.~Vardi]{Moshe Y. Vardi\rsuper{e}}	
	\address{Department of Computer Science, Rice University, USA}
	\email{vardi@cs.rice.edu} 	





	\begin{abstract}

The concept of \emph{decomposition} in computer science and engineering
is considered a fundamental component of \emph{computational thinking} and is
prevalent in design of algorithms, software construction, hardware design, and more.
We propose a simple and natural formalization of \emph{sequential decomposition},
in which a task is decomposed into two sequential sub-tasks, with the first sub-task to
be executed before the second sub-task is executed. These tasks are specified
by means of input/output relations.  We define and study \emph{decomposition problems},
which is to decide whether a given specification can be sequentially decomposed. Our main
result is that decomposition itself is a difficult computational problem. More specifically,
we study decomposition problems in three settings: where the input task is specified
explicitly, by means of Boolean circuits, and by means of automatic relations.
We show that in the first setting decomposition is NP-complete, in the second setting
it is NEXPTIME-complete, and in the third setting there is evidence to suggest that
it is undecidable. Our results indicate that the intuitive idea of decomposition as a
system-design approach requires further investigation. In particular, we show that
adding a human to the loop by asking for a decomposition hint lowers the complexity of
decomposition problems considerably.
\end{abstract}

	\maketitle


		\section{Introduction}

Over the past decade, it became apparent that the conceptual way of analyzing
problems through computer-science techniques can be considered as a general approach to
analysis, design, and problem solving, known as \emph{computational thinking}~\cite{Wing06,Wing11}.
A key element of this approach is the taming of complexity by decomposing a complex problem
into simpler problems. Quoting Wing~\cite{Wing06}: ``Computational thinking is using abstraction and
decomposition when attacking a large complex task or designing a large complex system.''
While abstraction helps to tame complexity by simplifying away irrelevant details of a complex
problem, decomposition helps to tame complexity by breaking down complex problems into simpler ones.
In fact, decomposition is a generic project-management technique: project managers quite often face
challenges that require decomposition --- a large project that is divided among team members to make
the problem less daunting and easier to solve as a set of smaller tasks, where team members work on  tasks
that are in their specific fields of expertise.  As computer scientists and engineers, the concept
of decomposition is prevalent in the design of algorithms, in software construction, in hardware design,
and so on.
For example, a classical paper in software engineering studies criteria to be used in decomposing systems
into modules~\cite{Parnas72}.
Yet, in spite of the centrality of the concept of decomposition in computational thinking,
it has yet to be studied formally in a general theoretical setting (see related work).
Such a study is the focus of this work.

There are many different types of decomposition that can be considered.  Based on her understanding
of the problem, the \emph{decomposer} has to make a decision on how to decompose a given problem,
for example, by meeting certain constraints on the size of the sub-problems, or constraints on
the way that solved sub-problems ought to be recomposed.  A simple and natural way of decomposition
is \emph{sequential} decomposition in which a task is decomposed into two sub-tasks, where the first
sub-task is to be carried out before the second sub-task can be executed. A formal model
for sequential decomposition is the subject of this work. We assume that the given problem is specified
by means of an input/output relation. It is widely accepted that such relations are the most
general way to specify programs, whether for terminating programs~\cite{Hoare69}, where input and output
should be related values, or for nonterminating programs~\cite{Pnueli77}, where input and output are
streams of values. The \emph{decomposition problem} is to decompose a given input/output relation $R$,
between an input domain $\In$ and an output domain $\Out$, into two relations $R_1$ and $R_2$, such that
$R$ can be reconstructed from $R_1$ and $R_2$ using relational composition
(defined in Section~\ref{sec:perlim}).
To avoid trivial solutions, where either $R_1$ or $R_2$ is the identity relation, we assume that the
intermediate domain, that is, the co-domain of $R_1$, which is also the domain of $R_2$, is specified.
Intuitively, specifying the intermediate domain amounts to constraining the manner in which the first
task can ``communicate'' with the second task. Such a restriction can be viewed as a form of
\emph{information hiding}, which is one of the major criteria for decomposition in~\cite{Parnas72}.
As we show, sequential decomposition is nontrivial only when the channel of communication between the
first and second task has a small ``bandwidth'', appropriately defined.

We study sequential decomposition in three settings:
\emph{explicit}, \emph{symbolic}, and \emph{automatic}.
In the explicit setting, the input/output relation $R$ is specified
explicitly. In the symbolic setting, the domains and $R$ are finite but too large to be specified explicitly,
so $R$ is specified symbolically as a Boolean circuit. In the automatic setting, the domains and $R$ may be infinite,
so the domains are specified by means of an alphabet, over which $R$ is specified by means of a deterministic
finite-state automaton.

Our general finding is that sequential decomposition, viewed as a computational
problem, is itself a challenging problem.  In the explicit setting, the decomposition problem is
NP-complete.  This escalates to NEXPTIME-complete for the symbolic setting.
For the automatic setting the decomposition problem is still open, but we provide evidence and conjecture that
it is undecidable. Specifically, we show that even a very simple variant of the automatic setting can be viewed as equivalent to the Positivity problem, whose decidability is well known to be open~\cite{Soi76,OW14b}.
We do show, however, that a ``strategic'' variant of the automatic setting, in which the required relations are described as transducers is in EXPTIME\@.
These findings, that decomposition is an intractable problem,  can be viewed as a
``No-Free-Lunch'' result, as it says that decomposition, which is a tool to combat complexity, is
itself challenged by computational complexity. This means that while decomposition is an essential
tool, the application of decomposition is an art, rather than science, and
requires human intuition.

As such, we explore decomposition with
``a human in the loop'', where the role of the human is to offer a hint,
suggesting one of the terms of the decomposition, and the role of the decomposition algorithm
is to check if a valid decomposition can be found based on the hint.
Table~\ref{tab:compTable} summarizes our complexity results with and without hints.
We show that in most cases decomposition with
a hint is easier than decomposition with no hint.




\begin{table}[!h]
	\centering
	\setlength{\tabcolsep}{6pt}
	\caption{Computational complexity of decomposition without and with hints.}\label{tab:qe}%
    \label{tab:compTable}
	\begin{tabular}{@{}llll@{}}
		\toprule
		& Without hints &  With hint $R_1$ &  With hint $R_2$	\\
		\midrule
		Explicit	& NP-complete
		& PTIME	& PTIME	\\
		Symbolic	& NEXPTIME-complete
		& $\Pi_3^P$ & $\Pi_3^P$				\\
		Automatic   & undecidabe?
		& EXPSPACE  & EXPSPACE		\\
	    Strategic   & EXPTIME
		& EXPTIME & PTIME    \\
		\bottomrule

	\end{tabular}
\end{table}


%
%

		\section{Related work}


In this paper we introduce a new framework to study \emph{decomposition}
of system-level specifications into component-level specifications. In contrast,
most existing approaches in software engineering focus on \emph{composition},
that is, developing systems from independently developed components, or proving
system-level properties from component-level properties, for example~\cite{FFG91}.
This compositional approach is a fundamental approach in computer science
to taming complexity.

We now give a few examples of the composition-based approach to system
development.  The \emph{contract-theory} approach of Benveniste et
al.~\cite{DBLP:conf/fmco/BenvenisteCFMPS07} uses assumption/guarantee
to describe both the implementation and the environment of a given
component; that is, the implementation of the component should work
with any other component whose implementation satisfies the
assumption. Components are then composed in a bottom up fashion to
produce a large-size system. The same approach is pursued also in the
\emph{interface theories} of de Alfaro et
al.~\cite{DBLP:conf/emsoft/AlfaroH01}, the \emph{modal automata} of
Larsen~\cite{DBLP:conf/avmfss/Larsen89}, and the \emph{I/O automata} of
Lynch~\cite{Lynch-tuttle88}.

Operations such as \emph{quotient} (adjoint of composition) or contextual equation solving~\cite{DBLP:conf/lics/LarsenX90,DBLP:journals/acta/FahrenbergL14}
allow one to synthesize a component that can be composed with another
one in order to refine a large size specification. This gets close
to the problem of decomposition with \emph{hints}, described below,
but is still focused on synthesizing an implementation.
The BIP approaches of Sifakis~\cite{DBLP:journals/software/BasuBBCJNS11}
proposes to specify complex systems by aggregating smaller
components using a very expressive algebra. Recently in~\cite{DBLP:conf/atva/ChengBCYJRBK11}, the approach has been extended to
synthesize the order of communication between components, but the
focus there is not on decomposing specification.
Similarly, Lustig and Vardi have shown~\cite{LV13}
how to synthesize reactive systems satisfying given linear-temporal
properties by composing components from a given library of reactive components.
Closer to our work is an approach studied in model-driven software
engineering.  The Fragmenta tool~\cite{DBLP:journals/scp/AmalioG15}
provides algebraic descriptions of models and fragments based on graph
morphism. While the tool eases the decomposition task, it does not
specifically handle the problem of actually providing such a decomposition.

A major weakness of composition-based approaches to system development is
that they consider only the problem of simplifying the implementation work,
but  ignore the task of decomposing the often very complex system-level
technical specifications into smaller/simpler ones. For example, there is
no technique to explain how smaller assumption/guarantee
contracts can be obtained from larger ones. This operation has to be
conducted manually by developers using their intuition and understanding of
complex system-level specifications.  Thus, our work here on decomposition
complements existing approaches on composition-based development.
In addition, we believe that our work
is also relevant to \emph{architectural design},
e.g., as a complement of~\cite{DBLP:conf/cbse/BliudzeSBJ14}.

A classical paper of Parnas in software engineering studies criteria
to be used in decomposing systems into modules~\cite{Parnas72}.
Parnas's framework, however, assumes that the starting point for
decomposition consists of an architectural specification of the system,
while our starting point is quite more abstract, as we assume that
we are provided with a relational specification.
Another related work is that of~\cite{Rath95}, which
describes an approach for extracting sequential components from a system
specification.  Unlike, however, our work here, which starts from a highly
abstract relational specification, the approach
of~\cite{Rath95} assumes that the system's specification is provided by means of an
interface specification of the components. Thus, this approach is more in the sense of a factorization rather then a decomposition.

Decomposition has been studied in the context of linear algebra.
A matrix decomposition or matrix factorization is a factorization of a matrix
into a product of matrices. There are many different matrix factorizations.
Certain Boolean matrix-factorization problems are known to be NP-complete~\cite{johnson87}. Our NP-completeness result for explicit relations can be
viewed as a special case of Boolean matrix-factorization. In addition, our  formulation for the explicit case can be viewed as a reformulation of the combinatorial definition of the nondeterministic communication complexity (see Chapters 1--2 in~\cite{Kushilevitz-Nisan}). In that sense, this paper extends these works to more general representations of relations.

			\section{Preliminaries}\label{sec:perlim}

\subsection{Relations}
Let $A,B,C$ be sets. For a binary relation $R\subseteq A\times B$, let $Dom(R)$, and $Img(R)$ be the domain of $R$, and the image (sometimes called co-domain) of $R$, defined as follows. $Dom(R)=\{a\in A\mid \exists b\in B\text{ s.t } (a,b)\in R\}$, and $Img(R)=\{b\in B\mid \exists a\in A\text{ s.t } (a,b)\in R\}$. For $a\in A$, let $Img_a(R)=\{b\in B\mid  (a,b)\in R\}$. The relation $R$ is called a function if for every $a\in A, b,b'\in B$ we have $(a,b),(a,b')\in R \implies b=b'$.
Given binary relations $R_1\subseteq A\times B$, and $R_2\subseteq B\times C$, the \textit{composition} of $R_1$ and $R_2$ is a binary relation $R_1\circ R_2\subseteq A\times C$ where $R_1\circ R_2=\{(a,c)\mid\exists b\in B\text{ s.t. } (a,b)\in R_1\text{ and } (b,c)\in R_2\}$.

%
%
\subsection{Automata}

 A Nondeterministic Finite Automaton (NFA) is a tuple  $\Aut=(\Sigma,Q,q,\delta, F)$, where $\Sigma$ is a finite alphabet, $Q$ is a finite state set with an initial state $q$, $\delta: Q\times\Sigma\to 2^Q$ is a transition function, and $F$ is an accepting-state set. A run of $\Aut$ over a word $w=a_1a_2\cdots a_n$ for some $n$ is a state sequence $r=q_0, q_1,\ldots ,q_n$ such that $q_{i+1}\in\delta(q_i,a_i)$ for $i\geq 0$, where $q_0=q$ is the initial state. A run is \textit{accepting} if its final state is accepting. A word is \textit{accepted} if it has an accepting run. The language of $\Aut$, $L(\Aut)$, is the set of all accepted words of $\Aut$ and is called a \textit{regular language}. A language is also regular if and only if it can be described by a \textit{regular expression}. We define the size of the automaton $\Aut$ as $|Q|+|\Sigma|+|\delta|$ and denote this size by $|\Aut|$.
For NFAs  $\Aut_1=(\Sigma_1,Q_1, q^1,\delta_1, F_1)$ and $\Aut_2=(\Sigma_2,Q_2,q^2,\delta_2, F_2)$, we define the \textit{product automaton} of $\Aut_1$ and $\Aut_2$ as the automaton $\Aut_1\times \Aut_2 = (\Sigma_1\times \Sigma_2,Q_1\times Q_2,q^1\times q^2,\delta, F_1\times F_2)$ where $(p,p')\in\delta((q,q'),(l,l'))$ iff $p\in\delta_1(q,l)$ and  $p'\in\delta_2(q',l')$.
An NFA $\Aut$ is deterministic (called DFA) if for every state $q$ and letter $a$, $|\delta(q,a)|\leq 1$. Every NFA can be determinized to a DFA that describes the same language by using the \textit{subset construction}, possibly with an exponential blow-up~\cite{HopcroftMotwaniUlman2003}.
It is often more convenient for users to specify regular languages by means of regular expressions, which can be converted to DFA, possibly with an exponential blow-up as well~\cite{HopcroftMotwaniUlman2003}. We assume here that all regular languages are specified by means of DFAs, as we wish to study the inherent complexity of decomposition.

Finally, a \textit{transducer} (we work here with Moore machines) is a deterministic finite-state machine with no accepting states, but with additional output alphabet and an additional function from the set of states to the output alphabet. Formally in our setting a transducer is a tuple $T=(\Sigma,\Sigma',Q,\delta,q_0, Out)$ where $\Sigma$ is a finite alphabet and $\Sigma'$ is a finite output alphabet, $Q$ is a finite set of states with an initial state $q_0\in Q$, $\delta:Q\times \Sigma\to Q$ is the transition function and $Out:Q\to\Sigma'$ is the output function. Transducers describe finite-state functions from input words in $\Sigma^*$ to output words in $\Sigma'^*$. We define $\delta(\varepsilon)=q_0$ and $T(\varepsilon)=Out(q_0)$. Inductively, $\delta(wa)=\delta(\delta(w),a)$ and $T(wa)=T(w) \cdot Out(\delta(wa))$, for $\in \Sigma^*$ and $a\in\Sigma$.

	\section{Problem definition}\label{sec:probdef}
The concept of decompositions that we explore here is related to systems that can be defined by their given input and produced output.
We model these as an input domain $\In$ and an output domain $\Out$, not necessarily finite. Our description of a system is
a specification that associates inputs to outputs, and is modeled as a relation $R\subseteq\In\times\Out$~\cite{Hoare69,Pnueli77}.
In addition we assume a constraint in form of a domain with a specific size that directs the decomposition to be more concise
and is given as an intermediate domain $\Barr$.
The objective is to decompose $R$ into relations $R_1\subseteq\In\times\Barr$, and $R_2\subseteq\Barr\times\Out$ such that the composition $R_1\circ R_2$ has either of the following properties: (i) no input-output association is added or lost --- this problem is called the \textit{Total Decomposition Problem (TDP)}, or (ii) a more relaxed version, called the \textit{Partial Decomposition Problem (PDP)} in which no input-output association is added, but we are allowed to lose some of the output as long as each input can be resolved.

 To make the paper more fluent to read we use the notation TDP/PDP for statements that are valid to the TDP and the PDP variants respectively.
	Since the size of the intermediate domain $\Barr$ can be significantly smaller than the size of the input or output domains, the problem becomes non-trivial as some sort of compression is required in order to solve the TDP/PDP\@. The actual problems of TDP/PDP are  appropriately defined for each section as decision problems.
	We first define the TD/PD conditions as follows.

		\begin{defi}\label{def:TDPPDP} (TD/PD conditions)
			Given binary relations $R\subseteq \In \times \Out$,
			$R_1\subseteq \In \times \Barr$, and $R_2\subseteq \Barr \times \Out$ for some domains $\In,\Out,\Barr$, we say that $(R_1,R_2)$ meets the TD condition if $Img(R_1)\subseteq Dom(R_2)$ and $R_1\circ R_2 = R$.
			We say that $(R_1,R_2)$ meets the PD condition if $Img(R_1)\subseteq Dom(R_2)$,  $Dom(R_1\circ R_2)=Dom(R)$, and $R_1\circ R_2 \subseteq R$.
		\end{defi}

%

		The decision problem of TDP/PDP, formally defined for diverse settings, is: given a description of domains and a relevant relation $R$, find whether there exist $(R_1,R_2)$ that meets the TD/PD condition.

		Since the type of domains and relations that we explore varies between an explicit and a more implicit description of relations, the formal definitions of the problem change according to these representations, and so are the sought decomposed relations. It is important to note that the TD/ PD  conditions are properties of the actual relations, and not of the description in which the relations are represented.

		Note that  PDP without the restriction of $Dom(R_1\circ R_2)=Dom(R)$ becomes trivial, as one can take the empty sets as $R_1$ and $R_2$.

		The following technical fact  can help prove TD/PD conditions in various settings.

		\begin{clm}\label{clm:premClm}
			Assume $Img(R_1)\subseteq Dom(R_2)$. Then  $Dom(R_1\circ R_2)=Dom(R_1)$ (and therefore $Dom(R)=Dom(R_1)$).
		\end{clm}

		\begin{proof}
			Let $i\in Dom(R_1\circ R_2)$. Then there is $o\in\Out$ such that $(i,o)\in R_1\circ R_2$. Which means there is $b\in\Barr$ such that $(i,b)\in R_1$ and $(b,o)\in R_2$. Specifically there is $b\in\Barr$ such that $(i,b)\in R_1$, therefore $i\in Dom(R_1)$.
			Next let $i\in Dom(R_1)$. Then there is $b\in \Barr$ such that $(i,b)\in R_1$. Since $Img(R_1)\subseteq Dom(R_2)$, we have that $b\in Dom(R_2)$ therefore there is $o\in\Out$ such that $(b,o)\in R_2$. Therefore since $(i,b)\in R_1$ and $(b,o)\in R_2$ then $(i,o)\in R_1\circ R_2$ which makes $i\in Dom(R_1\circ R_2)$.
		\end{proof}

		Finally see that the decomposed relations that meet the TD conditions, also meet the PD conditions on the same input, therefore a positive answer to TDP implies a positive answer to PDP\@.

		\section{Decomposition is hard}

Decomposition has been advocated as the first step in the design of complex systems, with the intuition
that it is easier to design components separately, rather than design a complex monolithic system.
We show in this section several settings in which finding a sequential decomposition
is computationally hard.  This means that while sequential decomposition could be used to simplify the
complexity of the initial specification, such decomposition itself is intractable, thus can be viewed
as a ``No-Free-Lunch''.

	\subsection{Explicit relations}\label{sec:Rel}

%
%

The simplest case of decomposition is when the domains $\In,\Out$ and $\Barr$ are finite and given explicitly as a part of the input, and the relation $R$ is given explicitly as a table in $\In\times\Out$.
	\begin{prob}\label{prob:ExpEDP}
		(TDP/PDP on explicit relations) We are given a tuple $I=(\In,\Out,\Barr, R)$ where $\In,\Out,\Barr$ are finite domains and $R\subseteq \In \times \Out$.
		The problem is whether there exist relations $R_1\subseteq \In \times \Barr$ and $R_2\subseteq \Barr \times \Out$ such that $(R_1,R_2)$ meet the TD/PD conditions.
	\end{prob}

	\begin{clm}\label{clm:small}
		If $|\Barr|\geq |\In|$ or $|\Barr|\geq |\Out|$ then TDP has a positive solution (and therefore PDP as well) and the relations that solve TDP can be found in a linear time to the size of the input.
	\end{clm}

	\begin{proof}
		Assume that $|\Barr|\geq |\Out|$ (the other case is analogue). Then we can define $R_2=\{(g(o),o)\mid o\in\Out\}$ where $g:\Out\to \Barr$ is any injection.
		Then for TDP, $R_1=\{(i,b)\mid (i,o)\in R \text{ and } b=g(o)\}$ is a relation that satisfies
		$R_1\circ R_2 = R$.  \end{proof}

%
%
%

	Therefore  TDP/PDP become non-trivial when $|\Barr|$ is strictly smaller than $|\In|$ and $|\Out|$.

	\begin{exa}
		Let $\In=\{i_1,i_2\}, \Barr=\{b\}$, and $\Out=\{o_1,o_2\}$. Let $R=\{(i_1,o_1), (i_2,o_2)\}$.
		Then the answer to TDP is negative as every non-empty composition of relations $R_1\circ R_2$ with $Dom(R_1\circ R_2)=Dom(R)$, must also include $(i_1,o_2)$ or $(i_2,o_1)$.
	\end{exa}

	We next show that even for the explicit setting, TDP/PDP are computationally hard, that is NP-complete.
	 On a positive note, being in NP, solutions for TDP/PDP can be sought by various techniques such as reduction to SAT, then using SAT solvers.

	\begin{thm}\label{thm:EDP}
		TDP/PDP on explicit relations are NP-complete.
	\end{thm}

\begin{proof}
	\textbf{TDP on explicit relations is NP-complete.}
	To see that TDP is in NP, guess $R_1,R_2$ and verify the TD conditions.
	We show hardness by a reduction from the NP-complete problem: Covering by Complete Bipartite Subgraphs (CCBS) (Problem GT18 at~\cite{GJ79}).
	In the CCBS we are given a bipartite graph $G$ and $k>0$, and the problem is whether $G$ can be covered by $k$ complete bipartite subgraphs.

	Given a CCBS instance $G=(V_1,V_2,E)$ and $k>0$ we define a TDP instance with $\In=V_1$ and $\Out=V_2$ (w.l.o.g.\ every vertex in $V_1$ and in $V_2$ has an incident edge). Set $R=E$ and add the $\Barr$ elements $\Barr=\{u_1,\cdots u_k\}$.
	Suppose there is a solution $(R_1,R_2)$ to TDP\@. For every $1\leq i\leq k$ let $V_1^i=\{v\in V_1| (v,u_i)\in R_1\}$, and let $V_2^i=\{v\in V_2\mid (u_i,v)\in R_2\}$. From the TD conditions it follows that $(v,v')\in R$ for every $v\in V_1^i, v'\in V_2^i$. Furthermore,  for every $(v,v')\in R$ there is $1\leq i\leq k$ such that $v\in V_1^i, v'\in V_2^i$. Since $E=R$, then $\{(V_1^i,V_2^i)\mid 1\leq i\leq k\}$ is a collection of complete bipartite subgraphs that covers $G$.

	Next, assume there is a solution $\{(V_1^i,V_2^i)\mid 1\leq i\leq k\}$ to the CCBS\@. Then define $R_1=\bigcup_{i\leq k}\{(v,u_i)\mid v\in V_1^i\}$ and $R_2=\bigcup_{i\leq k}\{(u_i,v)\mid v\in V_2^i\}$. Then $R_1\subseteq \In\times \Barr$, $R_2\subseteq \Barr\times \Out$ and $Img(R_1)\subseteq Dom(R_2)$. To see that $R_1\circ R_2=R$, let $(v,v')\in R_1\circ R_2$. Then there is $u_i$ such that $(v,u_i)\in R_1,(u_i,v)\in R_2$, so $(v,v')\in(V_1^i,V_2^i)$, therefore $(v,v')\in R$.
	On the other hand let $(v,v')\in R$. Since the solution to CCBS covers $G$, there is $i\leq k$ such that $(v,v')\in (V_1^i,V_2^i)$, so $(v,u_i)\in R_1$ and $(u_i,v')\in R_2$, hence $(v,v')\in R_1\circ R_2$.

	\vspace{2mm}

	\textbf{PDP on explicit relations is NP-complete.} Membership in NP is easily shown. We show hardness by reduction from Set Cover (Problem SP5 in~\cite{GJ79}). In Set Cover we are given a set of elements	$S=\{a_1,\cdots a_n\}$, a set of subsets of $S$, $C= \{c_1,\cdots c_m\}$ and $k\geq 0$. The problem is whether there are $k$ sets from $C$ whose union covers $S$. We assume w.l.o.g.\ that every member of $S$ belongs to at least one member of $C$.
	Given an instance $I=(S,C,k)$ of Set Cover,  define a PDP instance $I'$ where
	$\In=S$,  $\Out=C$,  $R=\{(a,c) \mid a\in c\}$ and add the intermediate domain elements to be $\Barr=\{b_1,\cdots b_k\}$.
	Assume that there is a set cover $\{c_{i_1},\cdots c_{i_k}\}\subseteq C$ for $S$. Then construct $R_1,R_2$ as follows. For every element $a$ and $j\leq k$, set $(a,b_j)\in R_1$ iff  $c_{i_j}$ is the first set for which $a\in c_{i_j}$. In addition
	set $(b_j,c_{i_j})\in R_2$. Then $R_1\subseteq \In\times\Barr$, $R_2\subseteq \Barr\times \Out$ and $Img(R_1)\subseteq Dom(R_2)$.
	Since the set cover for $S$ covers all elements of $S$, we have that $Dom(R_1\circ R_2)=Dom(R)$.
	To see that $R_1\circ R_2\subseteq R$, let $(a,c_{i_j})\in R_1\circ R_2$ for some $j\leq k$. Then $(a,b_j)\in R_1$ and $(b_j,c_{i_j})\in R_2$. By definition of $R_1$, $a\in c_{i_j}$, therefore we have that $(a,c_{i_j})\in R$.

	Next, assume  $I'$  has a PDP solution with relations $(R_1,R_2)$.
	Define a set cover as follows.
	For every $b_j\in Img(R_1)$ (and therefore $b_j\in Dom(R_2)$), choose a single element $c_{i_j}\in Img_{b_j}(R_2)$ and set $c_{i_j}\in C'$. Then $|C'|\leq k$.
	To see that $C'$ is a set cover for $S$, let $a\in\In$. Then since $a\in Dom(R)$ and $Dom(R_1\circ R_2)=Dom(R)$, there is
	$c\in \Out$ such that  $(a,c)\in R_1\circ R_2$. Therefore there is $b_j\in\Barr$ such that $(a,b_j)\in R_1$ and by definition $(b_j, c_{i_j})\in R_2$, which means $(a,c_{i_j})\in R_1\circ R_2$ as well.
	Since $R_1\circ R_2\subseteq R$ then $(a,c_{i_j})\in R$ which means $a\in c_{i_j}$. As $c_{i_j}\in C'$ our proof is complete. 	\end{proof}

		\subsection{Symbolic relations}\label{sec:symbrel}

In Section~\ref{sec:Rel}, the input relation is described explicitly.
In many cases, however, although finite, the relation is too large to be
described explicitly, and it makes more sense to describe it symbolically. Specifically,
the domains are given as the set of all truth assignments over sets of Boolean variables,
and the relation is described symbolically. Such representations have been studied in the literature,
where they are often referred to as \textit{succinct} representations, since they allow for a polynomial-size
description of exponential-size domains and relations.  In this section we explore a standard encoding
in which the relation is described as a Boolean circuit, as in~\cite{das2016cnf,balcazar1992complexity}.
Other symbolic encodings studied in the literature are Boolean formulas~\cite{Veith97},
and BDDs~\cite{FeigenbaumKVV99}.

In general, a \textit{succinct representation} (also called \textit{circuit description} in our setting) of a binary word $w$ is a boolean circuit that on input $i$ in binary emits two boolean values $x,y$ as an output: $x=1$ iff  $i\leq |w|$, and if $x=1$ then $y$ is the $i$'th bit of $w$~\cite{balcazar1992complexity}.
In that sense the circuit description of an integer $k$ is a boolean circuit that on input $i$ in binary emits $1$ iff $i\leq k$, and the
circuit description of a binary relation $R$ of natural numbers is a boolean circuit $C_R$ that on input $(i,j)$ in binary decides whether $(i,j)$ belongs to $R$.
As such, the circuit description of a language $A$ is defined to be the set of all circuits that succinctly describe words in $A$~\cite{balcazar1992complexity}.
For more about circuit description, see~\cite{GalperinW83,balcazar1992complexity,PapadimitriouY86}.



\begin{prob}\label{prob:TDPCircuit}
(TDP/PDP for symbolic relations) We are given a tuple $I=(n_\In, n_\Out, n_\Barr, C_R)$ where $n_\In, n_\Out, n_\Barr$ are natural numbers, and $C_R$ is a circuit description of a relation $R\subseteq\In\times\Out$ where $\In=\{0,1\}^{n_I}$ and $\Out=\{0,1\}^{n_O}$.
The problem is whether there exist relations  $R_1\subseteq \In \times \Barr$ and $R_2\subseteq \Barr \times \Out$,  where $\Barr=\{0,1\}^{n_\Barr}$, such that $(R_1,R_2)$  meet the TD/PD conditions.

\end{prob}


Note that since the domains $\In,\Barr, \Out$ are finite, the requirement for TDP/PDP for symbolic relation can equivalently be to find whether there exist circuits $C_{R_1},C_{R_2}$ that describe relations $R_1\subseteq \In \times \Barr$
and $R_2\subseteq \Barr \times \Out$, such that $(R_1,R_2)$  meet the TD/PD conditions.

For TDP/PDP, as in Claim~\ref{clm:small}, the problem becomes trivial when $n_\Barr \geq \min\{n_\In,n_\Out\}$. Note that a variant of PDP, in which the required relations in the solution are functions can be viewed as an instance of the problem of Boolean functional synthesis, e.g.~\cite{FTV16,CFTV18}.

We next show that TDP/PDP  are NEXPTIME-complete. We obtain this result by applying
the computational-complexity theory of succinct-circuit representations for
``logtime'' reductions~\cite{balcazar1992complexity} to the NP-hardness reductions described in
Section~\ref{sec:Rel}.  We describe this in details.

A reduction from a language $A\subseteq \Sigma^*$ to a language $B\subseteq\Sigma^*$, for a finite alphabet $\Sigma$,
is a function $f:A\rightarrow B$ such that $x\in A$ iff $f(x)\in B$. The function $f$ is called a \textit{logtime reduction}
if the $i$-th symbol of $f(x)$ can be computed in a time logarithmic in the size of $x$. This is done by using a so called
``direct-input-access'' Turing machine, which has a specific ``index'' tape in which a binary index~$i$ is written and
then the $i$-th symbol of the input string $x$ is computed. See~\cite{balcazar1992complexity}, which also states
a generalization of the following:

\begin{thm}\label{thm:balcazar} \emph{(Balcazar, Lozano, Toran~\cite{balcazar1992complexity})}
For every language $B\subseteq\Sigma^*$, if $B$ is NP-hard under logtime reducibility,
then the succinct representation of $B$ is \newline NEXPTIME-hard under polynomial-time reducibility.
\end{thm}

Since ``standard'' NP-hard problems (under polynomial-time reducibility) tend to be NP-hard also under logtime reducibility, the crux of Theorem~\ref{thm:balcazar} is that a succinct representation of a given problem in a form of a circuit does not necessarily ease the complexity of solving the problem.


From Theorem~\ref{thm:balcazar} we get:

\begin{thm}\label{thm:succ}
TDP/PDP for symbolic relations are NEXPTIME-complete.
\end{thm}

\begin{proof}

First see that an equivalent formulation of the TDP/PDP on symbolic relation is the following: \textit{We are given a circuit $C_I$ that describes a word that encodes a tuple $I=(n_\In, n_\Out, n_\Barr, R)$ where $n_\In, n_\Out, n_\Barr$ are natural numbers, and $R$ is a relation $R\subseteq\In\times\Out$ where $\In=\{0,1\}^{n_I}$ and $\Out=\{0,1\}^{n_O}$.
The problem is then whether there exist relations  $R_1\subseteq \In \times \Barr$
and $R_2\subseteq \Barr \times \Out$,  where $\Barr=\{0,1\}^{n_\Barr}$, such that $(R_1,R_2)$  meet the TD/PD conditions.}
Note that such a formulation is a succinct representation of TDP/PDP for explicit relation in which the domains are $\In=\{0,1\}^{n_\In}$, and $\Out=\{0,1\}^{n_\Out}$. Membership in NEXPTIME is easily shown since the explicit variant can be obtained from the circuit representation in exponential time and both explicit TDP/PDP are in NP\@.
In addition, as we next show, we can use this formulation to prove hardness by showing that each of the reductions in a chain of reductions from SAT to explicit TDP or to explicit PDP, is a logtime reduction, and then use Theorem~\ref{thm:balcazar} to show NEXPTIME hardness for succinct TDP/PDP\@.
For self containment of the paper we added in Appendix~\ref{sec:AppSymbolic} the definitions of the languages mentioned below along with the relevant reductions and references.

	For languages $A$ and $B$, we use the notations  $A\leq^P_m B$ to denote that $A$ is polynomial-time reduced to $B$.
	We have that $SAT\leq^P_m 3SAT\leq^P_m KCOLORABILITY\leq^P_m \newline PARTITIONCLIQUES \leq^P_m  CCBS  \leq^P_m TDP$, where  all the reductions are described in Appendix~\ref{sec:AppSymbolic} apart from the last reduction which is Theorem~\ref{thm:EDP} from this paper.
	Similarly, we have that $SAT\leq^P_m CLIQUE \leq^P_m VC \leq^P_m SC \leq^P_m   PDP$
	where all the reductions are described in Appendix~\ref{sec:AppSymbolic}, the last reduction is again from Theorem~\ref{thm:EDP}.

	As every problem $A$ is a set of binary strings (words) $x$, and as every reduction from $A$ to $B$ is a function $f$ such that $x\in A$ iff $f(x)\in B$, we have that the reductions described above are standard in the sense that only a few bits are required from $x$ in order to determine the identity of the $j$'th bit of $f(x)$. For example, in the reduction from $3SAT $ to $CLIQUE $, the output is a binary string that represents an adjacency matrix of a graph of size polynomial to the size of the input (and this size can be found in logarithmic time, see~\cite{balcazar1992complexity}). The content of every cell $(i,j)$ of that matrix can be determined by whether two satisfying assignments to 3SAT, one for a certain clause $c$ and one for a certain clause $c'$ in the input share the same variable (see~\cite{Karp72} for more details). In this reduction, the output formula is standard, in the sense that the location of the clauses $c,c'$ in the input string is easily found. As such, determining a certain output bit requires only a small (logarithmic) access to the input string, followed by a small (also logarithmic) comparison of the obtained input process.
	Similarly, the Cook-Levin reduction from any NP problem to $SAT$ is standard in that sense as well, and therefore the output as the SAT formula can be obtained from logarithmic input and comparison. See~\cite{GJ79} for more details.

	To show that the reductions from $CCBS$ to $TDP$ and from $SC$ to $PDP$ are logtime reductions, we need to be more accurate in the encoding of TDP/PDP\@. Indeed we can encode all problems as an adjacency matrix of the relation $R\subseteq \In\times\Out$ in addition to $k\geq 0$ (in unary) that describes the size of the intermediate domain $\Barr$. The input for $CCBS$ is encoded as an adjacency matrix and some $k\geq 0$. A standard encoding for $SC$ is also as an adjacency matrix (where the cell $(i,j)=1$ iff the element $a_i$ is a member of the set $c_j$) and some $k\geq 0$. As such, these reductions leave the input as is, since the difference between the instances of  TDP/PDP, $CCBS$ and $SC$ is only in the semantics.

	All in all  we have that (i) $SAT$ is NP-hard under logtime reductions, and (ii) every reduction described above is a logtime reduction. The Conversion Lemma from~\cite{balcazar1992complexity} states that if a language $A$ is logtime reduced to a language $B$ then the circuit representation of $A$ is polynomial-time reduced to the circuit representation of $B$.
	Therefore although the composition of two logtime reductions is not necessarily a logtime reduction, by using the Conversion Lemma from~\cite{balcazar1992complexity} we have that this chain of reductions for the succinct versions of the languages holds under polynomial-time reduction. Therefore by transitivity of polynomial-time reductions we get that TDP/PDP is NEXPTIME hard.

	As a side note, that may also serve as a suggestion as an alternative proog, see that the definitions of logtime reductions and the Conversion Lemma from~\cite{balcazar1992complexity} can be easily extended to poly-logtime reduction, that preserve transitivity and as can serve as an alternative proof.
\end{proof}

		\subsection{Automatic relations}\label{sec:interleave}

In many applications we need to consider input and output as streams of symbols, with some desired relation between the input stream and the output stream.  The most basic description  for such systems, the one that we explore in this work, is when the domains are (possibly infinite)  sets of finite words over finite alphabets, and the given relation is a regular language, given as a deterministic finite automaton (DFA), over the product alphabet of the input and output domains. The setting that we consider in this paper is of \textit{automatic relations}, which are relations that are described by automata as defined below. Automatic relations provide a context for a rich theory of automatic structures, cf.~\cite{KhousNer95,Rubin04}, with a  solvable decision for first-order logic. We follow here the convention in \emph{regular model checking}, cf.~\cite{bouajjani2000}, where \emph{length-preserving} automatic relations, with input and output symbols  interleaved, are used as input/output specifications for each step of reactive systems. (Thus, unlike the definitions in~\cite{KhousNer95,Rubin04}, we do not allow padding.)

Given finite alphabets $\Sigma,\Sigma'$ with domains $D\subseteq \Sigma^*$, $D'\subseteq \Sigma'^*$, and a
relation $R\subseteq D\times D'$, we say that a DFA $\Aut_R$ over the alphabet $(\Sigma\times\Sigma')$ describes $R$
if: $(\vec{d},\vec{d}')\in L(\Aut_R)$ if and only if $(\vec{d},\vec{d}')\in R$. Note the slight abuse of notation as
$L(\Aut_R)$ describes words in $(\Sigma\times\Sigma')^*$ while $R$ describes words in $\Sigma^*\times\Sigma'^*$.
Thus, we assume that input and output streams have the same lengths, that is, $(\vec{d},\vec{d}')\in R$ implies
that $|\vec{d}|=|\vec{d'}|$.


\begin{prob}\label{prob:EDPTSC}
(TDP/PDP for automatic relations) We are given a tuple \[I=(\Sigma_\In, \Sigma_\Barr, \Sigma_\Out, \Aut_R),\]
where  $\Sigma_\In$, $\Sigma_\Barr$, and $\Sigma_\Out$ are finite alphabets, and $\Aut_R$ is a DFA that describes
a relation $R\subseteq\Sigma_\In^*\times \Sigma_\Out^*$. The problem is whether there exist DFAs $\Aut_{R_1}, \Aut_{R_2}$
that describe relations $R_1\subseteq \Sigma_\In^*\times \Sigma_\Barr^*$ and $R_2\subseteq \Sigma_\Barr^*\times \Sigma_\Out^*$
such that $(R_1,R_2)$ meet the TD/PD conditions.
\end{prob}


As in the explicit and symbolic cases, the problem becomes non trivial for TDP/PDP only when $|\Sigma_\Barr| < \min\{|\Sigma_\In|, |\Sigma_\Out|\}$.
While TDP/PDP in the symbolic setting can be solved by reduction to the explicit setting, this cannot be done here
as the domains are possibly infinite.  Indeed, automatic-relation TDP/PDP seems to be a challenging problem.
We next conjecture that automatic-relation TDP/PDP is undecidable and
explain the motivation for this conjecture and why an automata-theoretic approach may not be helpful for TDP/PDP\@. We then show that even for the most basic case, in which the given relation is the equality relation, TDP/PDP can already be viewed as an algorithmic problem in automata that is equivalent to the Positivity Problem whose decidability is still open~\cite{Soi76,OW14b}.
Then, we show that for an intermediate alphabet that is of size of power of $2$, TDP/PDP on automatic relations can be reduced to TDP/PDP on binary intermediate alphabet.
Finally we show by an automata-theoretic approach that a ``strategic'' variant of PDP, in which the required relations are in form of transducers, is decidable, and in fact is in EXPTIME\@.



%

		\subsubsection{An undecidability conjecture for TDP/PDP}\label{sec:conj}

A notable positive result about automatic relations is the decidability of their first-order theories~\cite{KhousNer95,Rubin04}.
TDP/PDP are essentially second-order problems --- we ask for the existence of $R_1$ and $R_2$
under the TDP/PDP conditions.
Since there are second-order problems
over automatic relations that are known to be undecidable; for example, checking the existence of an Hamiltonian path in an automatic graph~\cite{KuskLohr10}, our conjecture is that this problem is undecidable as well.
We provide here intuition to justify this conjecture.

\begin{conj}\label{conj:only}
TDP/PDP for automatic relations is undecidable.
\end{conj}


To support the claim of how non-trivial the decomposition problem is, we consider a more simple and abstract variant of automatic TDP\@. We consider a very simple case in which the given automatic relation is trivial, and the intermediate alphabet is binary. We do not even require the decomposed relations $R_1,R_2$ to be realized by automata, although we do require the length of every matching words in $R_1$ and $R_2$ to be the same.
Specifically, given a regular language $L$ over $\Sigma_\In^*$, let $R^=_L\subseteq (\Sigma_\In\times\Sigma_\In)^*$
be the \textit{equality} relation over $L$, i.e. $R^=_L =\{(w,w) \mid w\in L\}$.
Given $L$ over $\Sigma_\In^*$ and $\Sigma_\Barr=\{0,1\}$, we ask whether there are relations $(R_1,R_2)$ that meet
the TD conditions with respect to $R^=_L$.

First see that we can assume w.l.o.g.\ that $(R_1,R_2)$ are functions since existence of such relations leads to existence of such functions.
Next, note that $R_1$ cannot relate two distinct words in $L$ to the same word in $\Sigma_\Barr^*$. This is because otherwise a word from
$\Sigma_\Barr^*$ has to meet two distinct $\Sigma_\In^*$ elements, thus either break the equality property or break the TDP property.
Therefore we have that $R_1$ is an injection from $L$ to $\Sigma_\Barr^*$. As such, finding an $R_1$ that is such an injection also gives us $R_2=R_1^{-1}$.

Let $L_n$ be the words in $L$ of size $n$.  Note that if there is $n$ for which $|L_n|>2^n$ then no such $R_1$ can be found.
If, however, for every $n$ we have $|L_n|\leq 2^n$ then a function $R_1$ can be simply realized by ordering the words
in every $L_n$ in lexicographic order and relating each one to the $\Sigma_\Barr^n$ word that encodes the index in binary.

Therefore the problem of TDP in this setting is reduced to the following problem: given a regular language $L$ over a finite alphabet $\Sigma$,
where for every $n$, $L_n$ is the set of words in $L$ of size $n$, does $|L_n|\leq 2^n$  for every $n$? We call this problem the
\emph{The Exponential-Bound Problem} (EBP).  The following theorem, however, shows that EBP is equivalent to the problem of \textit{Positivity}, described below, whose decidability has been famously open for decades~\cite{Soi76,OW14b}. By \textit {equivalent}, we mean that there is a many-to-one reduction from EBP to Positivity and vica versa.
Although this is not a direct reduction to automatic TDP, this relation indicates the hardness of
solving automatic TDP on even a simple relation.

A \emph{linear recurrence sequence} (LRS) is a sequence of integers $\langle u_n \rangle_{n=0}^{\infty}$
satisfying a recurrence relation: there exist integer constants $a_1, a_2, \ldots, a_d$ such that,
for all $n \geq 0$, $u_{n+d} = a_1u_{n+d-1} + a_2u_{n+d-2} + \cdots + a_d u_n$ (we say that such a sequence has \emph{order} $d$.)  If the initial values
$u_0, \ldots, u_{d-1}$ of the sequence are provided, the recurrence relation defines the rest of the
sequence uniquely.  Given a linear recurrence sequence (LRS) $\langle u_n \rangle_{n=0}^{\infty}$,
the \emph{Positivity Problem} asks whether all terms of the sequence are non-negative.
We next show that EBP and Positivity are equivalent. Thus the decidability of one problem implies the decidability of the other.

We begin by recalling a useful technical result about LRSs, whose proof can be found in~\cite{AAOW15}:

\begin{propC}[{\cite[Cor.~4]{AAOW15}}]%
\label{cor-4-proposition}
Let $\langle u_n \rangle_{n=0}^{\infty}$ be an integer LRS of order $d$. Then there exists a rational stochastic matrix $M$, of dimension $4d+5$, such that for all $n \geq 0$, we have
\[%
	\label{positivity-equation-prop}
	u_n \leq 0 \quad \mbox{iff} \quad
	(M^{2n+1})_{1,2} \leq \frac{1}{4}
	\, .
\]
Moreover, as noted in the comments following Corollary~4 in~\cite{AAOW15},
	$M$ can be chosen so that its entries are dyadic rationals, i.e.,
	having denominator some power of $2$.
\end{propC}

(Technically speaking, the proof of Corollary~4 in~\cite{AAOW15} constructs a dyadic-rational stochastic matrix $\widetilde{Q}$ such that
$u_n \leq 0$ iff $(\widetilde{Q}^{2n+1})_{1,4d+3} \leq 1/4$.
The desired matrix $M$ is then immediately obtained from $\widetilde{Q}$ by interchanging rows $2$ and $4d+3$, and interchanging columns $2$ and $4d+3$.)

We are now in a position to proceed with our equivalence:

\begin{thm}
	EBP is Equivalent to Positivity.
\end{thm}

\begin{proof}
	We first show that Positivity reduces to EBP\@. Let $\langle u_n
	\rangle_{n=0}^{\infty}$ be an LRS of order $d$; we show how positivity for the
	sequence $\langle -u_n \rangle_{n=0}^{\infty}$ can be formulated as an
	EBP problem. To this end, we invoke Prop.~\ref{cor-4-proposition}
	to
	obtain a stochastic matrix $M$ of dimension $4d+5$, all of whose entries are dyadic rationals, and such that, for all $n \geq 0$,
	\begin{equation}%
	\label{positivity-equation1}
	(M^{2n+1})_{1,2} \leq \frac{1}{4} \quad \mbox{iff} \quad u_n \leq 0
	\quad \mbox{iff} \quad -u_n \geq 0
	\, .
	\end{equation}
	In other words the positivity of $\langle -u_n \rangle_{n=0}^{\infty}$
	is violated iff there is some $n$ such that the $(1,2)$-\emph{th}
	entry of $M^{2n+1}$ is strictly larger than $1/4$.

	Let $2^p$ be the largest power appearing among the denominators of the entries of $M$. Write $J = 2^p M$
	and $N = (2^p M)^2$. Then $J$ and $N$ are square matrices with
	non-negative integer coefficients, and hence there is some DFA
	$\mathcal{A}$ such that $(J\cdot N^n)_{1,2}$ is the number of words of
	length $n+1$ accepted by $\mathcal{A}$. More precisely, $\mathcal{A}$
	has initial state $s$, and $4d+5$ further states $q_1, \ldots, q_{4d+5}$. The single
	accepting state is $q_2$. To define the transition function, if the
	$(1,j)$-\emph{th} entry of $J$ is $\ell$, then we postulate $\ell$
	transitions going from state $s$ to state $q_j$, each labelled with a
	new (fresh) letter. Likewise, if the $(i,j)$-\emph{th} entry of $N$ is
	$\ell$, then we include $\ell$ transitions going from $q_i$ to $q_j$,
	again for each one using a new letter as label. In this way, $J$ and $N$ can be
	viewed as the adjacency matrices of the underlying directed multigraph
	of $\mathcal{A}$, and $(J\cdot N^n)_{1,2}$ counts the number of paths
	in $\mathcal{A}$ going from $s$ to $q_2$ in $n+1$ steps. Since by
	construction, different paths give rise to different words, $(J\cdot
	N^n)_{1,2}$ does indeed correspond to the number of words of length
	$n+1$ accepted by $\mathcal{A}$.

	Writing $L(\mathcal{A}) = L$, Eq.~(\ref{positivity-equation1})
	becomes, for all $n \geq 0$,
	\begin{equation}%
	\label{positivity-equation2}
	L_{n+1} \leq \frac{2^p 2^{2pn}}{4} = \frac{2^{2p(n+1)}}{2^{p+2}}
	\quad \mbox{iff} \quad -u_n \geq 0
	\, .
	\end{equation}

	We now modify the automaton
	$\mathcal{A}$ by lengthening every transition in
	$\mathcal{A}$ by a factor of
	$2p$; more precisely, for every transition $q \rightarrow
	q'$, create $2p-1$ fresh non-accepting states $r_1, \ldots,
	r_{2p-1}$ and replace $q \rightarrow q'$ by the sequence
	$q \rightarrow r_1 \rightarrow r_2 \rightarrow \ldots \rightarrow
	r_{2p-1} \rightarrow
	q'$, all labelled with the same letter as the original transition. The
	initial and accepting states otherwise remain unchanged. Let
	us call the resulting DFA $\mathcal{A}'$, with accepted language
	$L(\mathcal{A}') = L'$. In moving from $\mathcal{A}$ to
	$\mathcal{A}'$, the net effect has been to increase the length of
	every accepted word by a factor of $2p$; note also that if $m$ is not a
	multiple of $2p$, then $L'_m = 0$.

	Combining the above with Eq.~(\ref{positivity-equation2}), we conclude that
	the LRS $\langle -u_n \rangle_{n=0}^{\infty}$ is
	positive iff for all $m \geq 0$, $L'_m \leq \frac{2^m}{2^{p+2}}$,
	i.e., $2^{p+2}L'_m \leq 2^m$.

	Inflating the alphabet size of $\mathcal{A}'$ by a factor of $2^{p+2}$,
	we can easily manufacture a DFA $\mathcal{A}''$ with accepted language
	$L(\mathcal{A}'') = L''$ having the property that, for all $m \geq 0$,
	$L''_m = 2^{p+2}L'_m$. It therefore follows that the LRS $\langle -u_n \rangle_{n=0}^{\infty}$ is
	positive iff for all $m \geq 0$, $L''_m \leq 2^m$, which
	completes the reduction of Positivity  to EBP\@.

	Finally, since the sequence obtained from the number of distinct words of length $n$
	 accepted by a given automaton is an LRS~\cite{OuaknineW15}, we have that EBP also reduces to Positivity, and therefore that the two problems are indeed equivalent.
	\end{proof}

		\subsubsection{Reduction to binary alphabet}\label{sec:red01}

 Since the size of the intermediate alphabet plays a crucial role in the solution of TDP/PDP, one may ask whether it suffices to search for solutions for only the binary case.
 We show that the answer is positive for intermediate alphabet that is of size of power or $2$. Specifically we show by reduction that for every automatic TDP/PDP instance $I=(\Sigma_\In,\Sigma_\Barr,\Sigma_\Out, \Aut_R)$,
where $|\Sigma_\Barr|=2^m$ for some $m>0$, there is a TDP/PDP instance $I'= (\Sigma_\In,\{0,1\},\Sigma_\Out, \Aut_{R'})$ for some relation $R'$
such that $I$ has a solution if and only if $I'$ has a solution.

To show our reduction,  we  first need to reason over automatic decomposition in the regular expressions representation level. Therefore we  give below  the following lemma that characterizes a possible TDP/PDP solution as regular expressions.


\begin{lem}\label{lem:char}
	A TDP/PDP for automatic relations instance $I=(\Sigma_\In,\Sigma_\Barr,\Sigma_\Out, \Aut_R)$ has a solution
	if and only if there is a regular expression $S\subseteq (\Sigma_\In\times\Sigma_\Barr\times \Sigma_\Out)^*$ such that the following happens:
	\begin{enumerate}
		\item (Only for the TDP variant) For every $(\vec{i},\vec{o})\in R$ there is $\vb\in\Sigma_\Barr^*$ such that $(\vec{i},\vb,\vec{o})\in L(S)$.
		\item (Only for the PDP variant) For every $\vec{i}\in Dom(R)$ there are $\vb\in\Sigma_\Barr^*$, and $\vec{o}\in\Sigma_\Out^*$ such that $(\vec{i},\vb,\vec{o})\in L(S)$.
		\item (For both variants) For every $(\vec{i},\vb,\vec{o}),(\vec{i}',\vb',\vec{o'})\in L(S)$, if $\vb=\vb'$ then $(\vec{i},\vec{o'}),(\vec{i}',\vec{o})\in R$.
	\end{enumerate}
\end{lem}

\begin{proof}
	Assume the instance $I$ has a TDP/PDP solution $(\Aut_{R_1},\Aut_{R_2})$. Let $\A'$ be the product automata of $\Aut_{R_1}$ and $\Aut_{R_2}$ followed by eliminating the edges $(i,b,b',o)$ where $b\not=b'$,  and then eliminating the first alphabet of $\Sigma_\Barr$. Therefore $\A'$ is over the alphabet $(\Sigma_\In\times\Sigma_\Barr\times\Sigma_\Out)$.
	Let $S$ be the regular expression that describes $\A'$.
	First verify (3): let $(\vec{i},\vb,\vec{o}),(\vec{i}',\vb,\vec{o'})\in L(S)$. Then  $(\vec{i},\vb),(\vec{i'},\vb) \in R_1$ and $(\vb, \vec{o}),(\vb,\vec{o'}) \in R_2$ which makes $(\vec{i},\vec{o'}),(\vec{i}',\vec{o})\in R_1\circ R_2$, hence both for TDP/PDP we have $(\vec{i},\vec{o'}),(\vec{i}',\vec{o})\in R$.
	For TDP we verify (1): Let $(\vec{i},\vec{o})\in R$, then there is $\vec{b}$ such that $(\vec{i},\vb)\in R_1$ and $(\vb, \vec{o})\in R_2$ therefore  $(\vec{i},\vb,\vec{o})\in L(S)$.
	For PDP we verify (2): let $\vec{i}\in Dom(R)$, then there is $\vec{o}$ such that $(\vec{i},\vec{o})\in R$. Then as before there is $\vec{b}$ such that $(\vec{i},\vb)\in R_1$ and $(\vb, \vec{o})\in R_2$ therefore  $(\vec{i},\vb,\vec{o})\in L(S)$.

	Next, assume $S$ is a regular expression over $(\Sigma_\In\times\Sigma_\Barr\times \Sigma_\Out)^*$ that meets conditions (1)--(3) for the TDP/PDP instance $I$. Let $\A'$ be the automaton that describes $S$. Let $\A_1$ be the automaton obtained from $A'$ by eliminating the $\Sigma_\Out$ letters and determinizing. Same let	$\A_2$ be the automaton obtained from $A'$ by eliminating the $\Sigma_\In$ letters and determinizing. Let $R_1,R_2$ be the regular relations that $\A_1,\A_2$ respectively describe. We see that $(R_1,R_2)$ meet the TD/PD conditions. 
	Let $\vb\in Img(R_1)$. Then there is $\vec{i}$ such that $(\vec{i},\vb)$ is an accepting word for $\A_1$, therefore there exists $\vec{o}$ such that $(\vec{i},\vb,\vec{o})$ is accepting for $\A'$. Therefore $(\vb,\vec{o})$ is accepting for $\A_2$ so $\vb\in Dom(R_2)$.
	For TDP let $(\vec{i},\vec{o})\in R$. Then from condition (1) there is $\vb\in\Sigma_\Barr^*$ such that $(\vec{i},\vb,\vec{o})\in L(S)$, so $(\vec{i},\vb,\vec{o})$ is accepting for $\A'$, therefore $(\vec{i},\vb)\in R_1$, and $(\vb, \vec{o})\in R_2$ so $(\vec{i},\vec{o})\in R_1\circ R_2$.
	For PDP, let $\vec{i}\in Dom(R)$. Then from condition (2) there are $\vb, \vec{o}$ such that $(\vec{i},\vb,\vec{o})\in L(S)$, so $(\vec{i},\vb,\vec{o})$ is accepting for $\A'$, therefore $(\vec{i},\vb)\in R_1$, and $(\vb, \vec{o})\in R_2$ so $\vec{i}\in Dom(R_1\circ R_2)$.
	Finally for both TDP/PDP assume $(\vec{i},\vec{o})\in R_1\circ R_2$. Then there is $\vec{b}$ such that $(\vec{i},\vb)\in R_1$, and $(\vb, \vec{o})\in R_2$. Then there are $\vec{i'}, \vec{o'}$ such that $(\vec{i},\vb,\vec{o'})$ and $(\vec{i'},\vb,\vec{o})$ are accepting for $\A'$ and therefore in $L(S)$. Then from condition (3), $(\vec{i},\vec{o}) \in R$. \end{proof}

Using  Lemma~\ref{lem:char}, we prove the following theorem.

\begin{thm}\label{thm:sec01}
	Given an instance $I=(\Sigma_\In,\Sigma_\Barr,\Sigma_\Out, \Aut_R)$ where $|\Sigma_\Barr|=2^m$ for some $m>0$, there is an instance $I'= (\Sigma_\In,\{0,1\},\Sigma_\Out, \Aut_{R'})$, constructed from $I$, such that $I$ has a TDP/PDP solution if and only if $I'$ has a TDP/PDP solution.
\end{thm}

\begin{proof}
The idea behind the reduction is to encode every letter in $\Sigma_\Barr$ in binary, and use this
encoding for the construction of $I'$. $\Aut_{R'}$ is therefore constructed from $\Aut_R$ by replacing
every edge labeled $(i,o)$ with a path of $m$ edges, each with the same label $(i,o)$.


   	  Let $bin$ be a function in which $bin(b)$ is a binary encoding of the letter $b\in\Sigma_\Barr$ in $m$ bits, and let $bin_j(b)$ be the $j$'th bit in $bin(b)$. Since $\Sigma_\Barr$ is of size $2^m$, the function $bin$ is a bijection.
   	  Assume that $I$ has a solution $(\Aut_{R_1}$,$\Aut_{R_2})$.
   	  We first construct an automaton $\Aut_{R'_1}$ over alphabet $(\Sigma_\In\times \{0,1\})$ from $\Aut_{R_1}$ by replacing every edge labeled $(i,b)$ with an $m$-edges path with edge labels $(i,bin_0(b)),\cdots (i,bin_{m-1}(b))$, thus such a path describes the word $(i^m, bin(b))$ (where $i^m$ is the letter $i$ concatenated $m$ times).
   	  Same, we construct an automaton $\Aut_{R'_2}$ over the alphabet $(\{0,1\}\times\Sigma_\Out)$ from $\Aut_{R_2}$ by replacing every edge labeled $(b,o)$ with an $m$-edges path with edge labels that  describes the word $(bin(b),o^m)$.
   	  Then since $(R_1,R_2)$ solve TDP/PDP for the instance $I$, we have that $(R'_1,R'_2)$ solve TDP/PDP for the instance $I'$.

   	  The other side of the proof is more involved since suppose we assume $(\Aut_{R'_1}$, $\Aut_{R'_2})$ solve $I'$. Then we cannot guarantee that these automata have an ``$m$-path'' structure that can reverse the construction we have just described.
   	  For that, we define by induction an homomorphism $h:(\Sigma_\In\times\Sigma_\Barr\times\Sigma_\Out)^*\to ({\Sigma_\In}^m\times \{0,1\}^m\times{\Sigma_\Out}^m)^*$ as follows.
   	  For every letter $(i,b,o)$, we define $h((i,b,o))=(i^m,bin(b),o^m)$. For a word $\sigma\in (\Sigma_\In\times\Sigma_\Barr\times\Sigma_\Out)^*$ and a letter $x\in (\Sigma_\In\times\Sigma_\Barr\times\Sigma_\Out)$ we define $h(\sigma x)=h(\sigma)h(x)$. By definition, $h$ is an homomorphism, note that $h$ is also an injection.

   	  Now assume that there is a TDP/PDP solution to $I'= (\Sigma_\In,\{0,1\},\Sigma_\Out, \Aut_{R'})$, and recall that $\Aut_{R'}$ is obtained from $\Aut_R$ by replacing every letter $(i,o)$ with a word $(i^m,o^m)$. Let $S'$ be the regular expression obtained from Lemma~\ref{lem:char} w.r.t. $I'$, and let $L(S')$ be the regular language of $S'$. Since  every word in $L(S')$ is a concatenation of strings of the form $(i^m,t,o^m)$ for some $i\in\Sigma_\In$, $t\in\{0,1\}^m$ and $o\in\Sigma_\Out$, we have that $L(S')\subseteq ({\Sigma_\In}^m\times \{0,1\}^m\times{\Sigma_\Out}^m)^*$. Moreover, since $bin$ is a bijection we have that for every such $(i^m,t,o^m)$, $(i,bin^{-1}(t),o)$ is well defined and  $h(i,bin^{-1}(t),o) =(i^m,t,o^m)$,
   	  therefore the homomorphism $h$ is onto $L(S')$. Since the pre-image under homomorphisms of a regular language is a regular language as well (see~\cite{Kozen}), we have that there is a regular language $L(S)$ of a regular expression $S$ over $(\Sigma_\In\times\Sigma_\Barr\times\Sigma_\Out)^*$ such that $L(S)=h^{-1}(L(S'))$.

   	  Since $S'$ satisfies conditions  (1)--(3) of Lemma~\ref{lem:char} w.r.t the instance $I'$ and since $h$ is a bijection from $L(S)$ to $L(S')$, it is easy to check that $S$ satisfies the conditions (1)--(3) of Lemma~\ref{lem:char} w.r.t the instance $I$, which means that $I$ has a TDP/PDP solution. 
\end{proof}

Theorem~\ref{thm:sec01} is stated only for an alphabet $\Sigma_\Barr$ whose size is some power of $2$. When this is not the case then the proof for Theorem~\ref{thm:sec01} does not hold. The reason is that if $m$ is such that $|\Sigma_\Barr|>2^m$, then some letters in $\Sigma_\Barr$ may not be represented via the translation function $bin$, thus the first part of the proof does not hold. If on the other hand $m$ is such that $|\Sigma_\Barr|<2^m$ then $bin^{-1}$ may ``add'' additional letters to $|\Sigma_\Barr|$, thus the second part of the proof does not hold.
Since we do not yet know the solution for this general case, we make the following conjecture.

\begin{conj}
	There is a reduction for TDP/PDP from an intermediate alphabet of arbitrary size to the binary alphabet.
\end{conj}

Corollary~\ref{cor:unary} shows that TDP/PDP on automatic relations with unary intermediate alphabet is solvable.

		\subsection{Strategic automatic PDP}\label{sec:AutFDP}

A specific version of PDP on automatic relations, that captures essential concepts in synthesis~\cite{KV01,FTV16,PnueliR90}, is when we require the solution to consist of strategies. For that, we define \textit{Strategic automatic PDP}, or Strategic PDP in short, as PDP on automatic relations in which the required relations $R_1$ and $R_2$ are functions (defined in Section~\ref{sec:perlim}) that can be represented by finite-state transducers $T_1$ and $T_2$, respectively.

Strategic PDP can be viewed as a game of incomplete information. Since the information ``flows'' in one direction, the key to proving decidability for this problem is to view the problem as a one-way chain communication of distributed synthesis from~\cite{KV01} to synthesize the required transducers. The construction here is rather technical, and we assume familiarity with the techniques developed in~\cite{KV00a,KV01}.

\begin{thm}\label{automaticFDP}
Strategic PDP is in EXPTIME\@.
\end{thm}

\begin{proof}
The proof for Theorem~\ref{automaticFDP} relies heavily on definitions and terminology from~\cite{KV01}. To ease the reading of the proof we bring the definitions from~\cite{KV01} in Appendix~\ref{sec:strategicApp}.

We first describe the intuition of the proof by providing a high-level outline. Note that the required transducer $T_1$ is a finite-state realization of a function $f_1: \Sigma_\In^+\rightarrow \Sigma_\Barr$, and the required transducer $T_2$ is a finite-state realization of a function $f_2: \Sigma_\Barr^+\rightarrow \Sigma_\Out$. We first consider trees of the form $\tau: \Sigma_\In^*\rightarrow \Sigma_\Barr\times\Sigma_\Out$.
We intend such a tree to represent simultaneously both $f_1$ and $f_2$; this requires that this tree satisfies additional constraints, discussed below. (Note that the label of the root is ignored here, as the empty word is not in the domain of $f_1$ and $f_2$.). A branch of this tree can be viewed as a word  $w\in (\Sigma_\In\times\Sigma_\Barr\times\Sigma_\Out)^\omega$. By running the DFA $\A_R$ for the automatic relation $R$ on $w$, ignoring $\Sigma_\Barr$, we can check that every prefix of $w$, viewed a finite word on $\Sigma_\In\times\Sigma_\Out$, is consistent with $R$. That is, a word in $(\Sigma_\In\times\Sigma_\Out)^+$
can viewed as a pair of equi-length words in $\Sigma^+_\In$ and $\Sigma^+_\Out$, which has to be in $R$. Thus, we can construct a deterministic automaton $\A_R^t$ on infinite trees that checks that all prefixes of all branches in $\tau$ are consistent with $R$.

Note, however, that in $\tau$ the $\Sigma_\Out$ values depend not on the $\Sigma_\Barr$ values but on the $\Sigma_\In$ values, while in $f_2$ the domain is $\Sigma_\Barr^+$. So it is possible that the $\Sigma_\Out$ values do not depend functionally on the $\Sigma_\Barr$ values. So next we consider a tree $\tau_2: \Sigma_\Barr^* \rightarrow \Sigma_\Out$. We can now simulate $\A_R^t$ on $\tau_2$ (cf.,~\cite{KV01}). The idea is to match each branch of $\tau$ with a branch of $\tau_2$ according to the $\Barr$-values. This means that we have several branches of the run of $\A_R^t$ running on one branch of $\tau_2$. We thus obtain an alternating tree automaton $\A_2$, whose size is polynomial in the size of $\A_R$. We can now check in exponential time non-emptiness of $\A_2$, to see if a function $f_2$ exists. If $\A_2$ is non-empty, then we can obtain a witness transducer $T_2$ whose size is exponential in $\A_R$.

Finally, we consider a tree $\tau_1: \Sigma_\In^*\rightarrow \Sigma_\Barr$ that represents $f_1$. We simulate both $\A_R$ and $T_2$ on each branch of $\tau_1$, where $T_2$ generates the $\Out$-values that were present in $\tau$ but not in $\tau_1$. Using these values $\A_R$ checks that each prefix is consistent with $R$. We have obtained a deterministic tree automaton $\A_1$, whose size is exponential in $\A_R$. We can now check non-emptiness of $\A_1$ in time that is linear in the size of $\A_1$, and obtain a witness transducer $T_1$. Thus, we can solve Strategic PDP for automatic relations in exponential time.

\vspace{2mm}
We now describe the proof of Theorem~\ref{automaticFDP} in detail. Recall that $R$ is the given automatic relation for the Strategic PDP problem. Let the DFA for $R$ be $\A_R=(\Sigma_\In\times\Sigma_\Out,Q,q_0,\delta,F)$; assume without loss of generality that $q_0\in F$ (since we do not care about the empty input word). We use $\A_R$ to construct  a deterministic B\"uchi tree automaton $\A_R^t$, over infinite trees  $\tau:\Sigma_\In^*\rightarrow\Sigma_\Barr\times\Sigma_\Out$, that checks that all prefixes of all branches in an input tree $\tau$ are consistent with $R$.
More precisely, let $a_1,\ldots,a_m$ be a nonempty sequence of input symbols from $\Sigma_\In$, and let $\tau(a_1 \ldots a_i)=(b_i,c_i)$ for $1\leq i \leq m$. Then we want the word pair $(a_1\cdots a_i, c_1 \cdots c_i)$ to be in $R$.


We define this tree automaton as follows:
$\A_R^t=(\Sigma_\In,\Sigma_\Barr\times\Sigma_\Out,Q',q_0,\delta^t,\alpha)$ is a $(\Sigma_\Barr\times\Sigma_\Out)$-labeled $\Sigma_\In$-tree deterministic automaton, with $Q'=\{q_0\}\cup (\Sigma_\In \times Q)$, and $\alpha = Q'$ as the acceptance condition. Intuitively, this tree automaton $\A_R^n$ carries with it the last letter in $\Sigma_\In$, to enable it to simulate $\A_R$.
For the transition function, we have as follows:
\begin{itemize}
\item
By slight abuse of notation, we define
$\delta^t(q_0)= \bigwedge_{i\in\Sigma_\In} \langle i,(i,q_0)\rangle$ (i.e., the successor state in direction $i$ is the state $(i,q_0)$).
(Note that in $\delta^t(q_0)$ there is no input letter, which means that we are ignoring the label of the root.)
\item
$\delta^t((i,q),(b,o))=\bigwedge_{i'\in\Sigma_\In} \langle i',(i',\delta(q,(i,o)))\rangle$, if $q\in F$, and $\delta^t((i,q),(b,o))=false$, if $q\not\in F$.
\end{itemize}
(The notation $\langle i',(i',\delta(q,(i,o)))\rangle$ means that the state $(i',\delta(q,(i,o)))$ is sent in direction $i'$.) Intuitively, $\A_R^t$ is a deterministic tree automaton that checks that all prefixes of all branches  of an input tree $\tau$ are consistent with $R$, by simulating $\A_R$ along all branches of $\tau$.
While the syntax used here is that of alternating automata, $\A_R^t$ send at most one successor state in each tree direction, so it is a deterministic automaton (Note that there is no disjunction in the definition of $\delta^t$, so there is no nondeterminism here.).
Note also that the size of $A_R^t$ is quadratic in the size of $A_R$.
If the automaton reach a state $q\not\in F$ then it has discovered an inconsistency with $R$, so it gets stuck and cannot accept the input tree.

Rather than use $A_R^t$, directly, it is more convenient to
convert it into a nondeterministic tree automaton $\A^n_R$ that guesses a tree $\Sigma_\Barr\times\Sigma_\Out$, and checks that it is accepted by $\A_R^t$. Formally,  $\A_R^n=(\Sigma_\In,Q,q_0,\delta^n,\alpha)$ is a nondeterministic automaton on unlabeled $\Sigma_\In$-trees\footnote{Formally, an unlabeled tree is the unique mapping from $\Sigma^*$ to $\{\bot\}$.}
with $\alpha = Q$ as the acceptance condition. The transition function is:
\begin{enumerate}
\item
$\delta^n(q_0)= \bigwedge_{i\in\Sigma_\In}
\bigvee_{b\in\Sigma_\Barr}\bigvee_{o\in\Sigma_\Out}
\langle i,\delta(q_0,(i,o))\rangle$
\item
$\delta^n(q)= \bigwedge_{i\in\Sigma_\In}
\bigvee_{b\in\Sigma_\Barr}\bigvee_{o\in\Sigma_\Out}
\langle i,\delta(q,(i,o))\rangle$, for $q\not=q_0$ if $q\in F$
\item
$\delta^n(q)= false$ for $q\not=q_0$ if $q\not\in F$.
\end{enumerate}
(Items (1) and (2) are identical, but we assumed $q_0\not\in F$.)
Note that an accepting run of $\A^n_R$ induces a
$(\Sigma_\Barr\times\Sigma_\Out)$-labeling of the $\Sigma_\In$-tree, due to the nodeterministic guesses in the transitions.

For a function $h:\Sigma_\In\rightarrow\Sigma_\Barr\times\Sigma_\Out$, let $h(i)_1, h(i)_2$ denote the first and second components of $h(i)$ respectively.
Then note that we can rewrite the first two clauses of the transition function as follows, where $\eta$ ranges over functions from $\Sigma_\In$ to  $\Sigma_\Barr\times\Sigma_\Out$.
\begin{enumerate}
\item
$\delta^n(q_0)=\bigvee_{\eta:\Sigma_\In\rightarrow (\Sigma_\Barr\times\Sigma_\Out)}
\bigwedge_{i\in\Sigma_\In}\langle i,\delta(q_0,(i,\eta(i)_2))\rangle$
\item
$\delta^n(q)=\bigvee_{\eta:\Sigma_\In\rightarrow (\Sigma_\Barr\times\Sigma_\Out)}
\bigwedge_{i\in\Sigma_\In}\langle i,\delta(q,(i,\eta(i)_2)\rangle$, for $q\not=q_0$.
\end{enumerate}


We next use $\A_R^n$ to construct the Strategic PDP solution transducers $T_1,T_2$ that describe trees $\tau_1,\tau_2$ respectively. We first construct $T_2$ by obtaining an alternating tree automaton $\A_2$. Every tree in $L(\A_2)$ is a partial solution for Strategic PDP\@.

We obtain $\A_2$ by simulating $A^n_R$ on $\Sigma_\Out$-labeled $\Sigma_\Barr$ trees, using an alternating automaton $\A_2=(\Sigma_\Barr,\Sigma_\Out,Q',q_0,\delta_2,\alpha_2)$, where $Q'=\{q_0\}\cup (Q\times\Sigma_\Out))$,
$\alpha_2=Q'$, and the transition function is defined as follows:
\begin{itemize}
\item
$\delta^2(q_0)=\
\bigvee_{\eta:\Sigma_\In\rightarrow (\Sigma_\Barr\times\Sigma_\Out)}
\bigwedge_{b\in\Sigma_\Barr}
\bigwedge_{i\in\Sigma_\In~\mbox{and}~b=\eta(i)_1}
\langle b, (\delta(q_0,(i,\eta(i)_2)),\eta(i)_2)\rangle$
\item
$\delta^2((q,o'),o)=
\bigvee_{\eta:\Sigma_\In\rightarrow (\Sigma_\Barr\times\Sigma_\Out)}
\bigwedge_{b\in\Sigma_\Barr}
\bigwedge_{i\in\Sigma_\In~\mbox{and}~b=\eta(i)_1}
\langle b, (\delta(q,(i,\eta(i)_2)),\eta(i)_2)\rangle$
if $q\in F$ and $o=o'$,
\item
$\delta^2((q,o'),o)=false$ if $q\not\in F$ or $o'\not=o'$.
\end{itemize}
In this simulation, every $(\Sigma_\Barr\times\Sigma_\Out)$-labeled branch of $\tau: \Sigma^+_\In\rightarrow \Sigma_\Barr\times\Sigma_\Out$ maps to a branch of $\tau_2: \Sigma^+_\Barr\rightarrow \Sigma_\Out$ with the same $(\Sigma_\Barr\times\Sigma_\Out)$-labeling.
Note, however, that in $\tau$ the $\Sigma_\Barr\times\Sigma_\Out$ letter pairs are all node labels, while in $\tau_2$ the $\Sigma_\Barr\times\Sigma_\Out$ letter pairs come from tree directions ($\Sigma_\Barr$) and node labels ($\Sigma_\Out$).

As shown in~\cite{KV01}, if $\tau_2$ is accepted by $A_2$, then an accepting run of $A_2$ on $\tau_2$ yields a tree $\tau$ where the $\Sigma_\Out$ labeling depends only on the $\Sigma_\Barr$ labeling.  We thus have to check that $\A_2$ is nonempty.  For this we use standard techniques in tree-automata theory~\cite{GTW02}, according to which nonemptiness of $\A_2$ can be checked in exponential time. A subtle point here is that the transition function of $\A_2$ is of exponential size, due to the disjunction $\bigvee_{\eta:\Sigma_\In\rightarrow (\Sigma_\Barr\times\Sigma_\Out)}$. But in the nonemptiness algorithm this disjunction of exponential size turns into an alphabet of exponential size, and this affects the computational complexity multiplicatively, so the algorithm is still exponential with respect to $\A_R$.

If the nonemptiness algorithm returns ``nonempty'', it returns also as a witness a finite-state transducer that generates a tree accepted by $\A_2$~\cite{GTW02}. The transducer is $T_2=(\Sigma_\Barr,\Sigma_\Out,S,s_0,\beta)$, where $S$ is a finite state set, $s_0\in S$ is the initial state, and $\beta:S\times\Sigma_\In\rightarrow S\times \Sigma_\Out$
is the transition function that yields, for a state $s\in S$ and input letter $b\in\Sigma_\Barr$, a pair $\beta(s,b)$ of a successor state and an output letter. The state set $S$ is, in general, of size that is exponential in $|Q|$.

We can now combine $A_R^n$ with $T_2$ to construct a nondeterministic tree automaton $\A_1$ for $\tau_1$, which is now an unlabeled $\Sigma_\In$-tree. Formally, $\A_1=(\Sigma_\In,Q\times S,(q_0,s_0),\delta_1,\alpha_1)$, where $\alpha_1=(Q\times S)$, and the transition function $\delta_2$ is obtained by combining $\delta^n$ with $\beta$:
\begin{enumerate}
\item
$\delta_2((q_0,s_0))=\bigvee_{\zeta:\Sigma_\In\rightarrow \Sigma_\Barr} \bigwedge_{i\in\Sigma_\In}\langle i,(\delta(q_0,(i,\beta(s_0,\zeta(i))_2),\beta(s_0,\zeta(i))_1))\rangle$
\item
$\delta_2((q,s))=\bigvee_{\zeta:\Sigma_\In\rightarrow \Sigma_\Barr} \bigwedge_{i\in\Sigma_\In}\langle i,(\delta(q,(i,\beta(s,\zeta(i))_2),\beta(s,\zeta(i))_1))\rangle$
\end{enumerate}
Nonemptiness of $\A_1$ can be checked in time that is linear in the size of $\A_1$~\cite{GTW02}, so it is exponential in the size of $\A_R$. Again, the nonemptiness algorithm  returns a finite-state transducer $T_1$ that induces a strategy mapping $\Sigma_\In^+$ to $\Sigma_\Barr$.

This concludes the proof that Strategic PDP can be solved in EXPTIME\@.	\end{proof}

		\section{Decomposition with a hint}\label{sec:hint}

Our results so far indicate that fully automated decomposition is a hard problem. Can we ameliorate
this difficulty by including a ``human in the loop''?  Indeed, in some decomposition scenarios,
a part of a decomposition is already given, and the challenge is to find the complementary component.
This can be thought of as a partial solution that is offered by a human intuition, e.g.\ a domain expert.
In the context of our framework we have that a candidate to either $R_1$ or $R_2$ is  given
(in the relevant description formalism) as a ``hint''. The question then is whether, given one possible component,
a complementary component indeed exists, and can be constructed, such that both relations together meet the TD/PD  conditions.

We discuss first TDP/PDP problems with a hint $R_1$ followed by a discussion with similar arguments for TDP/PDP problems with hint $R_2$.

\subsection{TDP/PDP with hint \texorpdfstring{$R_1$}{R1}}

TDP/PDP problem with a hint $R_1$ is formally defined as follows.

\begin{defi}\label{sec:defdef}
In the TDP/PDP problem with a hint $R_1$ we are given input, output and intermediate domain $\In,\Out,\Barr$,
relations $R\subseteq \In\times \Out$ and $R_1\subseteq \In\times \Barr$. The goal is to find whether there
is a relation $R_2\subseteq\Barr\times\Out$ such that $(R_1,R_2)$ meet the TD/PD conditions.
\end{defi}

The exact nature of the domains and relations varies as before according to the problem setting (explicit, etc.).
As we see, such a hint as a partial solution relaxes the computational difficulty of TDP/PDP,
shown in previous sections, considerably. To that end, we show the following \textit{maximum property}
that is relevant for all TDP/PDP settings.  Given a TDP/PDP instance with a hint $R_1$, define the relation
$R'_2\subseteq (\Barr\times\Out)$ to be
$R'_2 =\{(b,o) \mid \forall i\in \In( (i,b)\in R_1 \rightarrow (i,o)\in R)\}$. Note that $Dom(R'_2)$ can strictly contain $Img(R_1)$.

\begin{lem}\label{lem:MaxRel}
Every solution for TDP/PDP with a hint $R_1$ is contained in $R'_2$ and if there exists such a solution then $R'_2$ is a solution as well.
\end{lem}
\begin{proof}
We prove for TDP, we skip the proof for PDP that is almost identical.
Let $R_2$ be a solution to TDP with a hint $R_1$.
We first see that $R_2\subseteq R'_2$.
Let $(b,o)\in R_2$. Suppose there exists $i$ such that $(i,b)\in R_1$ and $(i,o)\not\in R$. Then we have that $R_1\circ R_2 \not\subseteq R$ which means that $R_2$ is not a solution, a contradiction. Therefore we have that  for all $i$,  if $(i,b)\in R_1$ then $(i,o)\in R$, hence $(b,o)\in R'_2$.
To see that $R'_2$ is a solution, first see that since $R_2$ is a solution contained in $R'_2$ then $Dom(R_1\circ R'_2) = Dom(R)$ and $Img(R_1)\subseteq Dom(R'_2)$.
Let $(i,o)\in R_1\circ R'_2$. Then there is a $b$ such that $(i,b)\in R_1$ and $(b,o)\in R'_2$, which means that by the definition of $R'_2$ and since $(i,b)\in R_1$, it must be that $(i,o)\in R$.
Finally,
assume that $(i,o)\in R$. Then there is $b$ such that $(i,b)\in R_1$, and $(b,o)\in R''_2$ for some solution $R''_2$ contained in  $R'_2$. Therefore $(b,o)\in R'_2$ so $(i,o)\in R_1\circ R'_2$.
\end{proof}

From Lemma~\ref{lem:MaxRel} we get a simple method to solve TDP/PDP with a hint $R_1$ for the explicit settings as follows: Construct $R'_2$ and check that $(R_1,R'_2)$
meet the TD/PD conditions. If $R'_2$ meets these conditions then $(R_1,R'_2)$ is a solution for TD/PD\@. Otherwise, there is no solution with $R_1$ as a hint.
From this observation we get:

\begin{thm}
TDP/PDP with a hint $R_1$ for explicit relations are in PTIME\@.
\end{thm}

The definition of $R'_2$ can also solve the symbolic setting with a hint. Thus the following result allows the use of QBF solvers for finding a solution for $R_1$.

\begin{thm}\label{thm:hintsucc}
TDP/PDP with a hint $C_{R_1}$ for symbolic relations are in $\Pi^P_3$.
\end{thm}
\begin{proof}
For the symbolic setting, when $R$ and $R_1$ are given as circuits $C_R$ and $C_{R_1}$, resp.,
we can use the definitions of $C_{R'_2}$ to verify the TD/PD conditions:
for $TDP$ verify the following formulas
(1) $\forall \vb, \vec{i}\exists\vec{o} (C_{R_1}(\vec{i},\vb)\rightarrow C_{R'_2}(\vb,\vec{o}))$, and
(2)  $\forall \vec{i},\vec{o} (\exists \vb (C_{R_1}(\vec{i},\vb) \wedge C_{R'_2}(\vb,\vec{o})))\leftrightarrow  C_R(\vec{i},\vec{o})$.
For PDP in addition to Formula (1) we also need to verify \newline
(3) $\forall \vec{i} (\exists \vec{b},\vec{o} (C_{R_1}(\vec{i},\vb)\wedge C_{R'_2}(\vb,\vec{o}))) \leftrightarrow  \exists \vec{o}' C_R(\vec{i},\vec{o}')$
, and (4) $\forall \vec{i},\vec{o} (\exists \vb (C_{R_1}(\vec{i},\vb) \wedge C_{R'_2}(\vb,\vec{o}))\rightarrow  C_R(\vec{i},\vec{o}))$.
In all these, $C_{R'_2}$ need not be constructed but can be described by the formula $(\forall \vec{i}) (C_{R_1}(\vec{i},\vb)\rightarrow C_R(\vec{i},\vec{o}))$.%
\footnote{ By performing quantifier elimination on this formula, $C_{R'_2}$ can be constructed explicitly, in time that may be exponential.\ c.f.,~\cite{FTV16}.}. \end{proof}

For the automatic relation setting we have the following result for TDP/PDP with a hint $R_1$, where
we assume that $R_1$ is given as a DFA\@.

\begin{thm}\label{thm:solveHint}
TDP/PDP with a hint  $\A_{R_1}$ for automatic relations are in EXPSPACE\@.
\end{thm}
\begin{proof}

	We start by recalling some basic automata constructions. Let $A,B,C$ be (possibly infinite) domains. If $\A$ is an automaton that describes a relation $R\subseteq A\times B$ (in a the sense of Section~\ref{sec:interleave}) then an NFA $Img(\A)$ that describes $Img(R)$ is obtained by discarding the $B$-alphabet labels from $\A$. Similarly an NFA $Dom(\A)$ that describes $Dom(R)$ is obtained by discarding the $A$-alphabet labels from $\A$. From automaton $\A_1$ that  describes a relation $R_1\subseteq A\times B$, and an automaton $\A_2$ that  describes a relation $R_2\subseteq B\times C$,  we obtain an automaton denoted by $\A_1\circ\A_2$ that  describes $R_1\circ R_2$ by taking the product automaton of $A_1$ and $A_2$, eliminate the labels $(a,b,b',c)$ in which $b\not=b'$, then discarding the $B$-alphabet labels from the remaining labels. Finally, to check if a DFA $\A$ is contained in a DFA $\A'$ we check for emptiness the automaton $\A\cap \neg \A'$.

Now for the proof, assume that $\Aut_R=(\Sigma_\In\times\Sigma_\Out,Q,q^0,\delta,F)$ is a DFA that describes $R$
and $\Aut_{R_1}=(\Sigma_\In\times\Sigma_\Barr,Q_1,q^0_1,\delta_1,F_1)$ is a DFA hint.
We show how to construct a DFA $\Aut_2 =(\Sigma_\Barr\times\Sigma_\Out,Q_2,q^0_2,\delta_2,F_2)$ that describes the maximum relation $R'_2$.
We define $\Aut'=(\Sigma_\In\times\Sigma_\Barr\times\Sigma_\In\times\Sigma_\Out,Q_1\times Q,(q_1^0,q^0),\delta_1\times\delta,F')$
as the product automaton of $\Aut_{R_1}$ and $\Aut_R$. We re-define the accepting states
$F'=F\times F_1$ by setting $F'=\{(q,q') \mid q\in F_1 \rightarrow q'\in F\}$.
For the transition function, we first discard from $\Aut'$ all the $(i,b,i',o)$ edges in which $(i\not=i')$.
Next we delete the alphabets $\Sigma_\In\times\Sigma_\In$ from all the labels in $\Aut'$. This results in a non-deterministic automaton,
that we determinize (with possibly an exponential blow-up) in the standard way through the subset construction, with the one exception
that we set every super-state (in the subset construction) to be accepting if and only if \textit{all} of its elements are accepting states.
This results in a DFA $\Aut_2$ over $\Sigma_\Barr\times\Sigma_\Out$ that describes the maximum relation $R'_2$.

Note that the construction of $\Aut_2$ can take exponential time in the size of $\Aut_{R_1}$ and $\Aut_R$ because of the determinization process.

\begin{clm}
$\Aut_2$ describes the maximum relation $R'_2$.
\end{clm}

\begin{proof}
Let $(\vec{b},\vec{o})\in L(\Aut_2)$.
Then for every $\vec{i}\in\Sigma_\In$ such that $(\vec{i},\vec{b})$ is in $L(\Aut_{R_1})$
we see that $(\vec{i},\vec{o})$ is in $L(\Aut_R)$, which means  $(\vec{b},\vec{o})\in R'_2$. Suppose otherwise, then we have that
$(\vec{i},\vec{b})$ reaches an accepting state in $\Aut_{R_1}$ and a non-accepting state in $\Aut_R$.
Therefore by our construction the word $(\vec{i},\vec{b},\vec{i'},\vec{o})$ is not accepting in $\Aut'$
before removing the $\Sigma_\In$ letters, therefore  $(\vec{b},\vec{o})$ is not in $\Aut'$, a contradiction.
Next, assume $(\vec{b},\vec{o})\in R'_2$.
Then for every $\vec{i}$ such that $(\vec{i},\vec{b})\in R_1$, we have that
$(\vec{i},\vec{o})\in R$. So if $(\vec{i},\vec{b})\in L(\Aut_{R_1})$ then $(\vec{i},\vec{o})\in L(\Aut_R)$.
Therefore for every $\vec{i}\in\Sigma^*_\In$, we have that $(\vec{i},\vec{b},\vec{i},\vec{o})$ is accepting
in the construction process of $\Aut_2$ before removing the $\Sigma_\In$ letters and determinizing.
So $(\vec{b},\vec{o})\in L(\Aut_2)$.
\end{proof}

Finally we need to check that $(\A_{R_1},\A_2)$ meet the TD/PD conditions. We do this step by step:
For TDP we need to check:
\begin{enumerate}
    \item $L(Img(\A_{R_1}))\subseteq L(Dom(\Aut_2))$ which means that $Img(\A_{R_1})\cap \neg Dom(\Aut_2) = \emptyset$. Note that both $Img(\A_{R_1})$ and $Dom(\Aut_2)$ are non-deterministic, but while checking the condition does not require determinization of $Img(\A_{R_1})$, to construct $\neg Dom(\Aut_2)$ we need to determinize $Dom(\Aut_2)$ with an extra exponential time, in the size of $\Aut_2$. We can however avoid direct construction by constructing the automata ``on-the-fly'' one state at a time. This leads to an overall construction in EXPSPACE\@.


\item $L(\Aut_{R_1}\circ\Aut_2) = L(\Aut)$. Here we need to check that $(\Aut_{R_1}\circ\Aut_2)\cap\neg\Aut = \emptyset$, which can be done in polynomial time since $\Aut$ is deterministic; and that $\neg(\Aut_{R_1}\circ\Aut_2)\cap\Aut=\emptyset$ that as before, requires determinization of $\Aut_{R_1}\circ\Aut_2$, which already has exponential size, all in all a double exponential blow up.
Again a indirect construction can be done as before in EXPSPACE\@.
\end{enumerate}
For PDP check that:
\begin{enumerate}
\item $L(Dom(\A_{R_1}))\subseteq L(Img(\Aut_2))$ which as before can be done in EXPSPACE\@.
\item $L(Dom(\Aut_{R_1}\circ\Aut_2)) = L(Dom(\Aut))$ which can also take exponential time, since as before checking $\neg Dom(\Aut_{R_1}\circ\Aut_2) \cap Dom(\Aut)=\emptyset$ requires determinization of $\Aut_{R_1}\circ\Aut_2$ which is already in exponential size. As before, this can be done in EXPSPACE\@.
\item  $L(\Aut_{R_1}\circ\Aut_2) \subseteq L(\Aut)$ which can be done in polynomial time.
Since $(\Aut_{R_1}\circ\Aut_2) \cap \neg \Aut=\emptyset$ since $\Aut$ is deterministic. \qedhere
\end{enumerate}

\end{proof}



\noindent
Along the lines of the proof of Theorem~\ref{thm:solveHint}, we can say the following on unary intermediate alphabet.

\begin{thm}\label{cor:unary}
	TDP/PDP for automatic relations with a unary intermediate alphabet is in PTIME\@.
\end{thm}

\begin{proof}
  For a unary intermediate alphabet $\Sigma_\Barr=\{b\}$, every word of length $n$ in $Dom(R)$ and in $Img(R)$  must be paired with the word $b^n$. Therefore if $(\vec{i},\vec{o})\in L(\Aut_R)$ for every $\vec{i}\in L(Dom(\Aut_R))$ and $\vec{o}\in L(Img(\Aut_R))$ then there is a TDP (and therefore PDP) solution constructed as follows: construct $\Aut_{R_1}$ from $\Aut_R$ by simply replacing every $\Sigma_\Out$ letter with the letter $b$ of $\Sigma_\Barr$,  and similarly construct $\Aut_{R_2}$ from $\Aut_R$ by simply replacing every $\Sigma_\In$ letter with the letter $b$ of $\Sigma_\Barr$.
  On the other hand if assume that there exists $\vec{i}\in L(Dom(\Aut_R))$ and $\vec{o}\in L(Img(\Aut_R))$ (suppose of size $n$) for which  $(\vec{i},\vec{o})\not\in L(\Aut_R)$. Then we have that no solution exists since then it means that $(\vec{i'},\vec{o}),(\vec{i},\vec{o'})\in L(\Aut_R)\in L(\Aut_R)$ for some $\vec{i'},\vec{o}$, but that means that $(\vec{i},b^n)$, $(b^n,\vec{o})$ must be accepting in any solution for $R_1,R_2$ respectively, then $(\vec{i},\vec{o})$ must be in $L(\Aut_R)$.

  Therefore the approach that we take, is to construct an NFA that is non-empty if and only if there indeed exists $\vec{i}\in L(Dom(\Aut_R))$ and $\vec{o}\in L(Img(\Aut_R))$ for which  $(\vec{i},\vec{o})\not\in L(\Aut_R)$, which means if and only if there is no TDP (or PDP) solution.
  For that, we construct an NFA $\Aut$ that is the product of $Dom(\Aut_R)$, $Img(\Aut_R)$ and $\Aut_R$.
  We then first delete every label $(i,o,i',o')$ for which $i\not=i'$ or $o\not= o'$. Then we set the accepting states of $\Aut$ such that every state $(q,q',q'')$ is accepting if and only if $q$ is an accepting state of $Dom(\Aut_R)$, $q'$ is an accepting state of $Img(\Aut_R)$, and $q''$ is a non accepting state of $\Aut_R$.

  To see that $\Aut$ is empty if and only if there is a TDP (or PDP) solution, assume first that $\Aut$ is non-empty. Therefore there is an accepting word $(\vec{i},\vec{o})$ which means that $\vec{i}\in L(Dom(\Aut_R))$ and $\vec{o}\in L(Img(\Aut_R))$ but $(\vec{i},\vec{o})\not\in L(\Aut_R)$, therefore no TDP (or PDP) solution exists. On the other hand if $\Aut$ is empty then it means that for every $(\vec{i},\vec{o})$ for which $\vec{i}\in L(Dom(\Aut_R))$ and $\vec{o}\in L(Img(\Aut_R))$, we have that $(\vec{i},\vec{o})\in L(\Aut_R)$ as well. Thus there is a TDP solution.

  Finally note that the construction of $\Aut$ is polynomial since no determinization is required. Thus since checking emptiness for NFA can also be done in polynomial time, we have that the problem is in PTIME\@.  \end{proof}

Finally, for Strategic PDP in the automatic setting, we assume that the hint is given in the form of a transducer.

\begin{thm}\label{thm:solveHintFDP}
Strategic PDP with a hint $T_1$ is in EXPTIME\@.
\end{thm}
\begin{proof}
For Strategic PDP with a hint $T_1$ as a transducer we simplify the construction of Section~\ref{sec:AutFDP}. There we constructed $A_2$ by taking simulating $A^n_R$ on $\tau_2$, and then checked non-emptiness for $A_2$ and obtained a witness transducer $T_2$. Then we plugged $T_2$ into $A^n_R$ and solved for $T_1$. Here we do not need to solve for $T_1$, because it is given as a hint. In addition, also recall that $A^n_R$ guesses mappings $\eta: \Sigma_\In\rightarrow \Sigma_\Barr\times\Sigma_\Out$. By taking the product of $A^n_R$ with $T_1$ we can ensure that these guesses are consistent with $T_1$.

Formally, let $T_1=(\Sigma_\In,\Sigma_\Barr,P,p_0,\gamma)$, where $P$ is a finite state set, $p_0\in P$ is the initial state, and $\gamma:P\times\Sigma_\In\rightarrow P\times \Sigma_\Barr$ is the transition function that yields, for a state $p\in P$ and input letter $i\in\Sigma_\In$, a pair $\gamma(p,i)$ of a successor state and an output letter.

Now we modify the construction of $\A_2$ from  Section~\ref{sec:AutFDP} and define \[\A'_2=(\Sigma_\Barr,\Sigma_\Out,Q'',(q_0,p_0),\delta'_2,\alpha'_2),\] where $Q''=\{q_0\}\cup (Q\times\Sigma_\Out)$,
$\alpha'_2=Q''$, and the transition function is defined as follows:
\begin{itemize}
\item
$\delta^2(q_0,p_0)=\
\bigvee_{\eta:\Sigma_\In\rightarrow \Sigma_\Out,}
\bigwedge_{b\in\Sigma_\Barr}
\bigwedge_{i\in\Sigma_\In~\mbox{and}~b=\gamma(p_0,i)_2}
\langle b, (\delta(q_0,(i,\eta(i))),\gamma(p_0,i)_1,\eta(i))\rangle$
\item
$\delta^2((q,p,o'),o)=
\bigvee_{\eta:\Sigma_\In\rightarrow \Sigma_\Out}
\bigwedge_{b\in\Sigma_\Barr}
\bigwedge_{i\in\Sigma_\In~\mbox{and}~b=\gamma(p,i)_2}
\langle b, (\delta(q,(i,\eta(i))),\gamma(p,i)_1,\eta(i))\rangle$ if \linebreak\mbox{$q\in F$} and $o=o'$,
\item
$\delta^2((q,p.o'),o)=false$ if $q\not\in F$ or $o'\not=o'$.
\end{itemize}
Now we continue as in~\ref{sec:AutFDP}; we check nonemptiness of $A'_2$, and if $\A'_2$ is nonempty, we get a transducer $T_2$ that defines a strategy from $\Sigma_\Barr^+$ to $\Sigma_\Out$. The complexity is exponential in the size of $\A_R$.
\end{proof}

\subsection{TDP/PDP with hint \texorpdfstring{$R_2$}{R2}}

We now discuss the case for TDP/PDP when $R_2$ is given as a hint. The definition, given below, is similar to that of $R_1$ being a hint.

\begin{defi}\label{def:defdefR2}
TDP/PDP problem with a hint $R_2$ we are given input, output and intermediate domain $\In,\Out,\Barr$,
relations $R\subseteq \In\times \Out$ and $R_2\subseteq \Barr\times \Out$. The goal is to find whether there
is a relation $R_1\subseteq\In\times\Barr$ such that $(R_1,R_2)$ meet the TD/PD conditions.
\end{defi}

As before, we  define the maximum relation $R'_1$ to be
$R'_1 =\{(i,b) \mid (\exists o (b,o)\in R_2) \wedge (\forall o( (b,o)\in R_2 \rightarrow (i,o)\in R))\}$. By an almost identical argument to that of Lemma~\ref{lem:MaxRel} we get the following.

\begin{lem}\label{lem:MaxR1}
	Every solution for TDP/PDP with a hint $R_2$ is contained in $R'_1$ and if there exists such a solution then $R'_1$ is a solution as well.
\end{lem}

For the explicit setting, construction again is easy. Therefore we have:

\begin{thm}
	TDP/PDP with a hint $R_2$ for explicit relations is in PTIME\@.
\end{thm}

And we also have

\begin{thm}
	TDP/PDP with a hint $C_{R_2}$ for symbolic relations is in $\Pi^P_3$.
\end{thm}

\begin{proof}
	The arguments for the proof are almost identical to those of Theorem~\ref{thm:hintsucc}, with the one change that this time we are
	given $C_{R_2}$ and we need to verify for TDP the formulas
	(1) $\forall \vb,\vec{i}\exists\vec{o} (C_{R'_1}(\vec{i},\vb)\rightarrow C_{R_2}(\vb,\vec{o}))$
	and (2) $\forall \vec{i},\vec{o} \exists \vb ((C_{R'_1}(\vec{i},\vb) \wedge C_{R_2}(\vb,\vec{o}))\leftrightarrow  C_R(\vec{i},\vec{o}))$.
	For PDP in addition to Formula (1) we also need to verify (3) $\forall \vec{i} (\exists \vec{b},\vec{o} (C_{R'_1}(\vec{i},\vb)\wedge C_{R_2}(\vb,\vec{o}))) \leftrightarrow  \exists \vec{o}' C_R(\vec{i},\vec{o}')$, and (4) $\forall \vec{i},\vec{o} (\exists \vb (C_{R'_1}(\vec{i},\vb) \wedge C_{R_2}(\vb,\vec{o}))\rightarrow  C_R(\vec{i},\vec{o}))$.
	In all these,
	\[C_{R'_1}\equiv (\exists\vec{o} C_{R_2}(\vb,\vec{o})\wedge \forall\vec{o} C_{R_2}(\vb,\vec{o})\rightarrow C_R(\vec{i},\vec{o})).\qedhere\] \end{proof}
%

Next, we can get a similar result to Theorem~\ref{thm:solveHint} for TDP/PDP for automatic relations with a hint $R_2$, but we need to be more careful since  $R_1$ imposes an extra condition of $Img(R_1)\subseteq Dom(R_2)$ that is needed to be considered.

\begin{thm}\label{thm:solveHintR2}
	TDP/PDP for automatic relations with a hint $\Aut_{R_2}$ is in EXPSPACE\@.
\end{thm}

\begin{proof}

    Assume that $\Aut=(\Sigma_\In\times\Sigma_\Out,Q,q^0,\delta,F)$ is a DFA for $R$,
	and that $\Aut_{R_2}=(\Sigma_\Barr\times\Sigma_\Out,Q_2,q^0_2,\delta_2,F_2)$ is a DFA for $R_2$ given as a hint. We show how to construct an automaton $\Aut_1 =(\Sigma_\In\times\Sigma_\Barr,Q_1,q^0_1,\delta_1,F_1)$, such that $L(\Aut_1)=L(R'_1)$ where $R'_1$ is the maximum relation from Lemma~\ref{lem:MaxR1}.

	We first construct an automaton that satisfies the second condition from the definition of $R'_1$, i.e.\ for every $\vec{o}\in \Sigma_\Out^*$, if $(\vb,\vec{o})\in R_2$ then $(\vec{i},\vec{o})\in R$.
	Similar to the proof of Theorem~\ref{thm:solveHint},
	we define $\Aut'=(\Sigma_\In\times\Sigma_\Out\times\Sigma_\Barr\times\Sigma_\Out,Q\times Q_2,(q^0,q^0_2),\delta\times\delta_2,F')$ to be the product automaton of $\Aut$ and $\Aut_{R_2}$. We redefine the accepting states to be $F'=\{(q,q') \mid q'\in F_2 \rightarrow q\in F\}$.
	To redefine the transition function, we first discard from $\Aut'$ all the edges $(i,o',b,o)$ edges in which $(o\not=o')$.
	Next we delete the alphabets $\Sigma_\Out\times\Sigma_\Out$ from all the labels in $\Aut'$. This results in a non-deterministic automaton, that we determinize in the usual way through the subset construction, with the one exception that we set every super-state (in the subset construction) to be accepting if and only if \textit{all} of its' elements are accepting states. This results in an automaton $\Aut'_1$ over $\Sigma_\In\times\Sigma_\Barr$.

	Next, we have to take care of the required condition that if $(\vec{i},\vb)$ is an accepting word in  $\Aut'_1$ then $\vb\in Dom(R_2)$.
	For that we first construct $\Aut'_2 = Dom(\Aut_{R_2})$ and determinize for an exponential blowup.
	Then, we define $\Aut_1$ by taking the product automaton of the DFAs $\Aut'_1$ and $\Aut'_2$, and discarding all edges labeled $(i,b,b')$ in which $b\not=b'$. Thus the overall construction has an expontential blowup.


	\begin{clm}
		$\Aut_1$ defines the maximum relation $R'_1$.

	\end{clm}

	\begin{proof}
		Let $(\vec{i},\vec{b})$ be an accepting word in $\Aut_1$. Then  $(\vec{i},\vec{b},\vec{b})$ is an accepting word in $\Aut'_1\times \Aut'_2$.
		Specifically $\vec{b}$ is an accepting word in $\Aut'_2$, therefore there is a word $\vec{o}\in\Sigma_\Out$ such that $(\vec{b},\vec{o})$ is accepting in $\Aut_{R_2}$ so $\vec{b}\in Dom(R_2)$.
		In addition since $(\vec{i},\vec{b})$ is an accepting word in $\Aut'_1$ then
		for every $\vec{o}\in\Sigma_\Out$ such that $(\vec{b},\vec{o})$ is in $L(\Aut_{R_2})$ we have that $(\vec{i},\vec{o})$ is in $L(\Aut)$. Otherwise we have that
		$(\vec{b},\vec{o})$ reaches an accepting state in $\Aut_{R_2}$, and a non-accepting state in $\Aut$. Therefore by our construction the word $(\vec{i},\vec{o},\vec{b},\vec{o})$ is not accepting in $\Aut'$ before removing the $\Sigma_\Out$ letters, therefore  $(\vec{i},\vec{b})$ is not in $\Aut'$, a contradiction.

		For the other side, assume $(\vec{i},\vec{b})\in R'_1$.
		Then $\vec{b}\in Dom(R_2)$, therefore $\vec{b}\in L(\Aut'_2)$.
		In addition, for every $\vec{o}$ such that $(\vec{b},\vec{o})\in R_2$, we have that
		$(\vec{i},\vec{o})\in R$. So if $(\vec{b},\vec{o})\in L(\Aut_{R_2})$ then $(\vec{i},\vec{o})\in L(\Aut)$. Therefore for every $\vec{o}\in\Sigma_\Out$, we have $(\vec{i},\vec{o},\vec{b},\vec{o})$ is accepting in the construction process of $\Aut'$ before removing the $\Sigma_\Out$ letters and determinizing.
		This means that $(\vec{i},\vec{b})\in L(\Aut')$.
		Therefore all in all we have that $(\vec{i},\vec{b})$ is accepting in $\Aut'_1$ and $\vec{b}$ is accepting in $\Aut'_2$, thus $(\vec{i},\vec{b})$ is accepting in $\Aut_1$ as required.
	\end{proof}

	All is left is to we verify that $(\Aut_1,\Aut_{R_2})$ meet the condition for TDP/PDP, which done in the same way as in the proof of Theorem~\ref{thm:solveHint}. Again, the verification requires constructing the complement  $\Aut_1\circ\Aut_{R_2}$ which requires exponential time as well as the complement of $Dom(\Aut_1)$. Since these constructions for verification purposes can also be done ``on-the-fly'' we conclude that the over all complexity of proving the theorem is in EXPSPACE\@.
\end{proof}

Finally for strategic PDP with a hint $T_2$ note that when $T_2$ is given as a hint we do not need to construct $\A_2$, but simply construct $\A_1$ from ${\A'}_R^t,T_2$ in the exact construction as above. Then we check for emptiness. As before, if $\A_1$ is empty then $T_2$ is not a good witness. Otherwise, we get a transducer $T_1$ such that $(T_1,T_2)$ meet the Strategic PD conditions. Note that this construction yields $T_1$ in polynomial time therefore we overall have the following.

\begin{thm}\label{thm:strategicPDPT2}
	Strategic PDP with a hint $T_2$ is in PTIME\@.
	\end{thm}
\begin{proof}
For Strategic PDP with a hint $T_2$ as a transducer we again simplify the construction of Section~\ref{sec:AutFDP}.  Now there is no need to construct the automaton $\A_2$. Instead, we are given a a finite-state transducer $T_2$, and we incorporate it in the construct of $A^t_R$.

Formally, let $T_2=(\Sigma_\Barr,\Sigma_\Out,P,p_0,\gamma)$, where $P$ is a finite state set, $p_0\in P$ is the initial state, and $\gamma:P\times\Sigma_\Barr\rightarrow P\times \Sigma_\Out$ is the transition function that yields, for a state $p\in P$ and input letter $i\in\Sigma_\Barr$, a pair $\gamma(p,i)$ of a successor state and an output letter.

We now plug in $T_2$ into the automaton $A^t_R$. Formally,
${\A'}_R^t=(\Sigma_\In,\Sigma_\Barr\,Q'',(q_0,p+0),\delta^t,\alpha)$ is a $\Sigma_\Barr$-labeled $\Sigma_\In$-tree deterministic automaton, with $Q'=\{(q_0,p_0)\}\cup (\Sigma_\In \times Q\times Q)$, and $\alpha = Q''$ as the acceptance condition.

For the transition function, we have as follows:
\begin{itemize}
\item
$\delta^t(q_0)= \bigwedge_{i\in\Sigma_\In} \langle i,(i,q_0,p_0)\rangle$
\item
$\delta^t((i,q,p),(b))=\bigwedge_{i'\in\Sigma_\In} \langle i',(i',\delta(q,(i,o),p'))\rangle$, where $o=\gamma(p,b)_1$ and $p'=\gamma(p,b)_2$, for $q\not=q_0$ and $q\in F$.
\item
$\delta^t((i,q,p),(b))=false$, for $q\not=q_0$ and $q\not\in F$.
\end{itemize}

\noindent
Note that ${\A'}_R^t$ is of size that is polynomial in $\A_R$ and $T_2$. Since non emptiness of deterministic tree automata can be checked in polynomial time~\cite{GTW02}, the claim follows.
\end{proof}

\subsection{Back to the undecidability conjecture}\label{sec:conj2}

%

%

In Section~\ref{sec:conj} we conjectured that TDP/PDP for automatic relations is undecidable. In this section we showed however that  TDP/PDP are decidable for automatic relations
when given a hint. Can we leverage this idea towards solving the problem of undecidability in full? That is, can we search for, say, an automaton for $R_1$ that together with the maximal $R_2$ constructed,
forms a solution for TDP/PDP\@?

It is tempting to try to use an automata-theoretic approach similar to the strategic PDP (see Section~\ref{sec:AutFDP}) in which we consider
representing $R_1 \subseteq (\Sigma_\In \times \Sigma_\Barr)^*$ as a labeled tree
$\tau_1: (\Sigma_\In\times\Sigma_\Barr)^* \mapsto \{0,1\}$. Then we can try to define
a tree automaton $\A_1$ that accepts a tree $\tau_1$ iff
it is a correct hint for the total or partial decomposition of an input/output
relation $R\subseteq (\Sigma_\In\times\Sigma_\Barr)^*$. The difficulty is that such an automaton
has to check two properties of $\tau_1$: (1) The domain of $R_1$ has to be
equal to the domain of $R$. This is essentially a requirement on the projection of $\tau_1$
on $\Sigma_\In$. (2) The composition of $R_1$ with the maximal $R_2$ is contained
in or equal to $R$. This is essentially a requirement on the projection of $\tau_1$
on $\Sigma_\Barr$. Known automata-theoretic techniques exist for dealing with projection of trees,
see for example~\cite{KupVar00}. We currently, however, have no known technique to deal with two orthogonal projections, as we have here.

Of course, this argument only shows that a particular technique to attack the problem is unlikely
to be successful. Therefore to further justify the conjecture, we note that the dual-projection problem
is reminiscent to the problem of distributed temporal synthesis,
which was shown to be undecidable~\cite{PnueliR90} (though there is no obvious formal connection
between the decomposition problem and the distributed-synthesis problem).
There we are representing an overall system strategy as a tree where tree labels
correspond to actions of system components and tree edges correspond to environment actions.
When different system components are expected to act independently, without knowledge of
actions by other components, the overall systems strategies has to be decomposed into
separate strategies as projections of the system strategies.

		\section{Discussion}

We studied here a formal model of sequential decomposition, a fundamental
concept in computational thinking. We showed that while decomposition is
viewed as an approach to tame design complexity, complexity is not
so easily tamed and decomposition can be quite difficult
when viewed as a computational problem. Human intuition, used to offer hints to the decomposition algorithm, is therefore necessary
to tame the complexity of the decomposition problem.
The complexity of TDP/PDP in the automatic-relation setting, conjectured to be undecidable, is still open. It is related to other decision problems for automatic relations, cf.~\cite{KuskLohr10}, and is a
subject of future work.

The decomposition problems studied in this paper are specific to relational decomposition, which arguably are enough to encode many types of decomposition, for example a type of decomposition cover where the subproblems overlap, or where the subproblems are required to satisfy certain properties that would greatly reduce the search space.  In that sense, as a further step towards exploring other types of decomposition, another line of future work is to explore the orthogonal problem of  \emph{parallel} decomposition, in which a task is decomposed into two tasks executed in parallel.
This subject has been studied in~\cite{AbadiL94}, but the specification of the system there
is given by the conjunction of its components' specifications, while we study specifications in
terms of overall system input/output relations.

Finally a more practical approach to sequential decomposition can be found in synthesis for symbolic settings. Indeed we showed in Section~\ref{sec:symbrel} that symbolic TDP/PDP is NEXPTIME-complete, but this stands only when the size of the intermediate domain $\Barr$ is given. For the synthesis approach we relax this requirement so we get to choose $\Barr$. On one hand the problem of solving TDP/PDP becomes uninteresting. On the other, however, The challenge is to find a TDP/PDP solution $(R_1,R_2)$ that allows good synthesis for a given relation $R$ in the sense that the composition of a separate synthesis to $R_2$ and $R_2$ can be a synthesis solution to $R$. A recent work inspired by this approach can be found in~\cite{CFTV18}.

			\bibliographystyle{alphaurl}
		\bibliography{DecompMergeRef,ref}


			\clearpage
		\appendix
		
\section{The reduction flow in Symbolic relations}\label{sec:AppSymbolic}

For the article to be self contained, we describe the problems and the chain of reductions that appear in the proof of Theorem~\ref{thm:succ}. Each problem is followed by it's index number in~\cite{GJ79}. The reduction are all from~\cite{Karp72}, apart from the last one which is from~\cite{orlin1977} and its location is stated in details below.

\paragraph{The list of problems:}
\begin{itemize}
    \item SAT (LO1):  The language of all satisfiable CNF Boolean formulas with $U$ variables and $C$ clauses.
    \item 3SAT (LO2): The language of all satisfiable CNF Boolean formulas with $U$ variables and $C$ clauses where each clause $C$ contains $3$ literals.
    \item KCOLORABILITY (GT4): The language of all graphs $G=(V,E)$ with $k>0$ for which there exists a coloring function $f:V\leftarrow\{1,\ldots,k\}$ such that $f(u)\not=f(v)$ for every edge $(u,v)\in E$.
    \item PARTITIONCLIQUES (GT15): The language of all graphs $G=(V,E)$ with $k\leq |V|$ for which there exists a partition of $V$ into $V_1,\ldots,V_k$ disjoint sets such that for every $1\leq i\leq k$, the graph induced by $V_i$ is a clique.
    \item CCBS (GT18): (Covering by Complete Bipartite Subgraphs):
    The language of all bipartite graphs $G=(V,E)$ with a positive integer $K\leq |E|$ for which there are $k\leq K$ subsets $V_1,\ldots,V_k$ of $V$ such that for every $1\leq i\leq k$, the graph induced by $V_i$ is a complete bipratite subgraph of $G$ and such that for every edge $(u,v)\in E$, there is some $V_i$ such that $(u,v)$ is contained in $V_i$.
    \item CLIQUE (GT19): The language of all graphs $G=(V,E)$ with $k\leq |V|$ for which there is a subset $V'\subseteq V$ of size $k$ or more such that every two vertices in $V$ are joint by an edge in $E$.
    \item VC (GT1): The language of all graphs $G=(V,E)$ with $k\leq |V|$ for which there is a subset  $V'\subseteq V$ of size $k$ or less such that for every edge $(u,v)\in E$ either $u\in V$ or $v\in V$.
    \item SC (SP5): The language of all $(S,C,K)$ where $C$ is a collection of all subsets of a finite set $S$ and $K\leq |C|$ such that $C$ contains a cover of $S$ of size $K$, i.e.\ there is a subset $C'\subseteq C$ of size $K$ or less such that every member in $S$ belongs to at least one member in $C'$.
\end{itemize}

\vspace{3mm}

\paragraph{The reductions:}

\begin{itemize}
    \item $SAT\leq_m^P 3SAT$: Every clause $(\ell_1,\ldots,\ell_j)$ in $C$ where the $\ell_i$'s are literal, is converted to a set of clauses $(\ell_1 \vee \ell_2\vee a_1)\wedge(\neg a_1 \vee \ell_2\vee a_2)\wedge\ldots (\neg a_{k-3} \vee \ell_{k-1}\vee \ell_k)$.

    \item $3SAT \leq_m^p Kcolorability$:
    Let $F$ be a $3SAT$ formula with variables $u_1,\cdots u_n$ and clauses $C_1,\cdots C_r$. Set $k=n+1$. Then $G-(V,E)$ is as follows. \newline $V=\{v_1,\ldots, v_n,v'_1,\ldots,v'_n,x_1,\cdots x_n,c_1,\cdots c_r\}$.
    To construct $E$ set $(v_i,v'_i)$ for every $i$, $(x_i,x_j)$,$(v_i,x_j)$ and $(v'_i,x_j)$ for every $i\not=j$, and finally add $(v_j,c_i)$ if $u_j$ does not appear in $C_i$ and add $(v'_j,c_i)$ if $\neg u_j$ does not appear in $C_i$. This reduction appeared in~\cite{SchurmannS03} as well.


    \item $Kcolorability \leq_m^p PartitionCLIQUE$:
    Set $G'$ to be the complement of $G$. Then $(G,k)\in Kcolorability$ iff $(G',k)\in PartitionClique$.

     \item $PartitionCLIQUE \leq_m^p CCBS$:
     Given a graph $G=(V,E)$ with $v_1,\ldots,v_n$ vertices and a non negative integer $k\leq n$, construct a bipartite graph $G'=(U,E')$ in which $U=u_1,\cdots u_n,w_1,\cdots w_n$ are vertices, and in addition for every $i\not=j$ add vertices $u_{ij},w_{ij}$. Construct $E'$ by setting edges $(u_i,w_i)$ for every $i$, and for every $i\not=j$, \[(u_i,w_j),(u_{ij},w_{ij}),(u_{ij},w_j),(w_{ij},u_i).\] Then $(G,k)\in PartitionClique$ if and only if $(G',\ell)\in CCBS$ where $\ell=k+|E|-|V|$.
     The reduction is shown in~\cite{orlin1977} as a proof of Theorem $8.1$ where $PartitionClique$ is described there as Problem $P_0$ and CCBS is described as Problem $P_2$.
       \item $SAT \leq_m^p CLIQUE$: Given a CNF formula $F$ with $C1,\cdots C_p$ clauses, we construct a graph $G=(V,E)$ where $V=\{<\ell,i>\mid \ell \text{ is a literal and occurs in } C_i\}$, and $E=\{<\ell,i>,<m,j>\mid i\neq j \wedge \ell\neq \neg m\}$.
    Then $F\in SAT$ if and only if $(G,p)\in CLIQUE$.

    \item $Clique\leq_m^P VC$: Given a graph $G=(V,E)$ and a positive integer $k\leq |V|$ we denote $G'$ to be the complement of $G$ and $k'=|V|-k$. Then $(G',k')\in VC$ iff $(G,k)\in Clique$.

    \item $VC\leq_m^P SC$:
    Given a graph $G=(V,E)$ and a positive integer $k\leq |V|$, let $C$ be the collection of all sets of edges $S_i$ incident on vertex $i$ for every $i\in V$. Then $(E,C,k)\in SC$ iff $(G,k)\in VC$.

\end{itemize}

		\section{Supplementary definitions for Strategic PDP}\label{sec:strategicApp}

       Since the proof of Theorem~\ref{automaticFDP} lies heavily on definitions and terminology from~\cite{KV01}, for containment of the paper we cite the definitions  and theorems from~\cite{KV01} that are of importance for the proof.

	   Given a finite set $\Upsilon$, a \textit{$\Upsilon$-tree} is a set $T\subseteq \Upsilon^*$ such that if $x\circ v\in T$ where $x\in\Upsilon^*$ and $v\in\Upsilon$ then also $x\in T$. When $\Upsilon$ is clear from the context, we call $T$ a \textit{tree}.
		When $T=\Upsilon^*$ we say that $T$ is \textit{full}. The elements of $T$ are called $\textit{nodes}$, and the empty word $\epsilon$ is the \textit{root} of $T$. For every $x\in T$, the nodes $x\circ v\in T$ where $v\in\Upsilon$ are the \textit{children} of $x$. Each node $x$ of $T$ has a direction, $dir(x)$ in $\Upsilon$. The direction of $\epsilon$ is $v^0$, for some designated $v^0\in\Upsilon$, called the \textit{root direction}. The direction of a node $x\circ v$ is $v$.

		Given two finite sets $\Upsilon$ and $\Sigma$,  a $\Sigma$-labeled $\Upsilon$-tree is a pair $(T,V)$ where $T$ is a $\Upsilon$-Tree and $V:T\rightarrow \Sigma$ maps each node of $T$ to a letter in $\Sigma$. When $\Upsilon$ and $\Sigma$ are not important or clear from the context, we call $(T,V)$ a labeled tree.

		 For a $\Sigma$-labeled $\Upsilon$-tree $(\Upsilon^*,V)$, we define the \textit{x-ray} of $(\Upsilon^*,V)$, denoted $xray((\Upsilon^*,V))$ as the $(\Sigma\times\Upsilon)$-labeled $\Upsilon$-tree $(\Upsilon^*,V')$ in which each node is labeled by both its direction and its labeling in $(\Upsilon^*,V)$. Thus, for every $x\in\Upsilon^*$, we have $V'(x) = (dir(x),V(x))$.

			An (alternating) tree-automaton that runs over full $\Sigma$-labeled $\Upsilon$-tree is a tuple $\Aut=(\Upsilon,\Sigma,Q,q_0,\delta,\alpha)$, where $\Sigma,\Upsilon$ are finite alphabets, $Q$ is a finite set of states, $q_0$ is an initial state, $\alpha\subseteq Q^\omega$ is an acceptance condition, and $\delta:Q\times \Sigma \to B^+(\Upsilon\times Q)$ is the transition function, where $B^+(\Upsilon\times Q)$ is the set of positive Boolean formulas over $\Upsilon\times Q$. When the automaton is in state $q$ and reads the label $a\in\Sigma$ of a node $x$, the Boolean formula $\delta(q,a)$ is a Boolean constraint on the states that the automaton transitions to on the children $x\cdot b$ of $x$, for $ b\in\Upsilon$.

			A special case, in which case the automaton is \emph{nondeterministic}, is when the automaton sends at most one successor state in each tree direction. Formally, this means that when considering a Boolean transition condition $\delta(q,a)$, for $q\in Q$ and $a\in\Sigma$, in disjunctive normal form, there is no disjunct (which is a conjunction) that contains two distinct atoms $(b,q_1)$ and $(b,q_2)$ associated with the same direction $b\in\Upsilon$. If $\delta(q,a)$ is nondeterministic, but contains a single disjunct, then it is \emph{deterministic}.

			A \textit{run} of an alternating automaton $A$ on an input $\Sigma$-labeled $\Upsilon$-tree $(\Upsilon^*,V)$ is a tree $(T_r,r)$ in which the nodes are labeled by elements of $(\Upsilon^*\times Q)$. Each node of $T_r$ corresponds to a node of $\Upsilon^*$. A node in $T_r$, labeled $(x,q)$, describes a copy of the automaton that reads the node $x$ of $T$ and visits the state $q$. The labels of the node and its children have to satisfy the transition function.
			Each infinite path $\rho$ in $(T_r,r)$ is labeled by a word $r(\rho)$ in $Q^\omega$. Let $inf(\rho)$ denote the set of states in $Q$ that appear in $r(\rho)$ infinitely often. A run $(T_r,r)$ is accepting iff all its infinite paths satisfy the acceptance condition. In \textit{Rabin} alternating automata, $\alpha\subseteq 2^Q\times 2^Q$, and an infinite path $\rho$ satisfies an acceptance condition $\alpha=\{(G_1,B_1),\cdots (G_k,B_k)\}$ iff there exists $1\leq i\leq k$ for which $inf(\rho)\cap G_i\not=\emptyset$ and $inf(\rho)\cap B_i=\emptyset$. The number of pairs in $\alpha$ is called the \textit{index} of $\alpha$. An automaton \textit{accepts} a tree if there is an accepting run on \emph{B\"{u}chi} automata, $\alpha\subseteq Q$, and an infinite path $\rho$ satisfies an acceptance condition if $inf(\rho)\cap F\not=\emptyset$.
			The language of the automaton $A$ is denoted $L(A)$. The automaton $A$ is said to be non-empty if $L(A)\not=\emptyset$. The size of an alternating automaton is the length of its description, which is typically dominated by the description of the transition function.

			We now bring Theorem $4.2$ from~\cite{KV01} as follows.
			\newline\textbf{Theorem 4.2 from~\cite{KV01}:} Given an alternating tree automaton $A$ over $(\Upsilon\times\Sigma)$-labeled $\Upsilon$-trees, we can construct an alternating tree automaton $A'$ over $\Sigma$-labeled $\Upsilon$-trees such that $A'$ accepts a labeled tree $(\Upsilon^*,V)$ iff $A$ accepts $xray((\Upsilon^*,V))$ and the automata $A'$ and $A$ have essentially the same size and index, i.e., within a constant factor. We call $A'$ the {cover of $A$ and is denoted by $cover(A)$.

\end{document}